\documentclass[a4paper,english]{lipics-v2021}

\input{macro.sty}

\newcommand{\alp}{A}
\newcommand{\alpo}{B}

\let\epsilon\varepsilon



\title{Reversible Transducers over Infinite Words}

\author{Luc Dartois}{Université Paris Est Creteil, LACL, F-94010 Créteil, France}{luc.dartois@u-pec.fr}{https://orcid.org/0000-0001-9974-1922}{}

\author{Paul Gastin}{Université Paris-Saclay, ENS Paris-Saclay, CNRS, LMF, 91190, Gif-sur-Yvette, France \and CNRS, ReLaX, IRL 2000, Siruseri, India}{paul.gastin@lmf.cnrs.fr}
{https://orcid.org/0000-0002-1313-7722}{}

\author{Loïc {Germerie Guizouarn}}{Université Paris Est Creteil, LACL, F-94010 Créteil, France}{loic.germerie-guizouarn@u-pec.fr}{https://orcid.org/0000-0002-3843-5427}{}

\author{R. Govind}{Uppsala University, Sweden}{govind.rajanbabu@it.uu.se}{https://orcid.org/0000-0002-1634-5893}{}

\author{Shankaranarayanan Krishna}{Indian Institute of Technology Bombay, Mumbai, India}{krishnas@cse.iitb.ac.in}{https://orcid.org/0000-0003-0925-398X}{}

\authorrunning{L. Dartois, P. Gastin, L. Germerie Guizouarn, R. Govind and S. Krishna}
\nolinenumbers

\hideLIPIcs  

\begin{document}
\maketitle

\begin{abstract}
Deterministic two-way transducers capture the class of regular functions. The efficiency of composing two-way transducers has a direct implication in algorithmic problems related to reactive synthesis, where transformation specifications are converted into equivalent transducers. These specifications are presented in a modular way, and composing the resultant machines simulates the full specification. 
An important result by Dartois et al.~\cite{DFJL17} shows that composition of two-way transducers enjoy a polynomial composition when the underlying transducer is reversible, that is, if they are both deterministic and co-deterministic. This is a major improvement over general deterministic two-way transducers, for which composition causes a doubly exponential blow-up in the size of the inputs in general. Moreover, they show that reversible two-way transducers have the same expressiveness as deterministic two-way transducers. However, the question of expressiveness of reversible transducers over infinite words is still open. 

In this article, we introduce the class of reversible two-way transducers over infinite words and show that they enjoy the same expressive power as deterministic two-way transducers over infinite words. This is done through a non-trivial, effective construction inducing a single exponential blow-up in the set of states. Further, we also prove that composing two reversible two-way transducers over infinite words incurs only a polynomial complexity, thereby providing foundations for efficient procedure for composition of transducers over infinite words.
\end{abstract}

\begin{gpicture}[name=automaton, ignore]
	\gasset{Nw=10,Nadjust=w,Nadjustdist=2,Nh=6,loopwidth=5,loopheight=6}
	\unitlength=1.1mm
	\node[Nmarks=i,linecolor=black](1)(0, 0){$1$}
	\node(2)(17, 10){$2$}
	\node(3)(17,-10){$3$}

	\drawedge(1,2){$b: 1$}
	\drawedge(2,3){$a: 1$}
	\drawedge[ELdist=-8, ELpos=35](1,3){$a: 1$}
	\drawloop[loopangle=0](3){$a: 0, b: 0$}
\end{gpicture}

\section{Introduction}
Transducers extend finite state automata with outputs. While finite state automata are computational models for regular languages, transducers are computational models 
for transformations between languages. Finite state automata remain robust in their expressiveness accepting regular languages across various descriptions
like  allowing two-way-ness, non-determinism and otherwise. However,  
this is not the case with transducers. Non-deterministic transducers realize relations
while deterministic transducers realize functions. Likewise, two-way transducers 
are strictly more expressive than one-way transducers: for instance, 
the function  reverse which computes the reverse of all input words 
in its domain is realizable by deterministic two-way transducers, but not by one-way transducers.  

One of the cornerstone results  of formal language theory is the beautiful connection 
which establishes that the class of regular languages corresponds to those recognized by finite state automata, to the class of languages definable in MSO logic, and to the class of languages whose syntactic monoid is finite.  Engelfriet and Hoogeboom 
\cite{EH01} generalized this correspondence between machines, logics and algebra in the case of regular languages to regular transformations. They showed that  regular transformations are those which are captured by two-way transducers and by Monadic second-order (MSO) transductions a la Courcelle~\cite{Cou94}. 
Inspired by this seminal work of Engelfriet and Hoogeboom, there has been an increasing interest over recent years in  characterizing the class of functions defined by deterministic two-way transducers~\cite{AlurFreilichRaghothaman14,BR-DLT18, DGK-lics18}. 

One such characterization is that of \emph{reversible} two-way transducers \cite{DFJL17} over finite words. Reversible transducers are those which are deterministic and also co-deterministic.  While determinism says that  any given state, on any given input symbol, 
does not transition to two distinct states, co-determinism says that no two distinct states can transition to  the same state on any input symbol. Reversibility makes the composition  
operation in two-way transducers very efficient : the composition of reversible transducers has polynomial state
complexity. This makes reversible transducers a very attractive formalism in synthesis where  specifications are given as relations of input-output pairs. \cite{DFJL17} showed that reversible two-way transducers capture the class of regular transformations. However, reversible transducers over infinite words have not been studied.

In another line of work, \cite{AFT12} initiated the study of  transformations on  infinite words. They considered functional, copy-less streaming string transducers (SST) with a M\"uller acceptance condition. 
An SST is a one-way automaton with registers; the outputs of each transition are stored in registers as words over the register names and the output alphabet.   In a run, the contents of the registers are composed. The  M\"uller acceptance condition is defined as follows: in any accepting run which settles down in a M\"uller set, the output is defined as a concatenation $x_1,x_2, \dots x_n$ 
of registers  where only $x_n$ is updated by appending words to $x_n$.  \cite{AFT12} 
proved the equivalence of this class of SST to deterministic two-way 
transducers with M\"uller acceptance and having an $\omega$-regular look-ahead. They also showed that these are equivalent to MSO transductions over infinite words. The $\omega$-regular look-aheads were necessary to obtain the expressiveness of MSO transductions on infinite words.  

In this paper, we continue the study of two-way transducers over infinite words. We introduce 
two-way transducers  with the parity acceptance condition ($\dbt$).
The main result of our paper is a  non-trivial generalization of \cite{DFJL17}, where we show that $\dbt$'s can be made reversible, obtaining $\rbt$ (two-way reversible transducers with parity acceptance).   Our conversion of $\dbt$ to $\rbt$ incurs a single exponential blow-up, and goes via   a new kind of SSTs that we introduce, namely,  copyless  SST with a parity acceptance condition ($\cbsst$). The parity condition used in both machines employs a  finite set of coloring functions, $c_1, \dots, c_k$, where each $c_i$ assigns  to the transitions of the underlying machine, a natural number. An infinite run $\rho$ is accepting if the minimum number which appears infinitely often  is even in all the $c_i$'s.

\begin{enumerate}
	\item We first show that starting from a 
	$\dbt$, we can obtain an equivalent $\cbsst$ where the number of states 
	of the $\cbsst$ is exponential in the number of states and coloring functions    
	of the  $\dbt$. The proof of this is a fairly non-trivial generalization of the classical Shepherdson construction~\cite{Shepherdson59} which goes from two-way automata to one-way automata.
	
	\item Then we show that, starting from a $\cbsst$ $\A$, we can obtain an equivalent 
	$\rbt$ $\mathcal{B}$ which is polynomial in the number of states and registers 
	of $\A$. This construction is a bit technical : we show that $\mathcal{B}$ is obtained as the composition of a deterministic one-way parity transducer $\mathcal{D}$  and an  $\rbt$ $\mathcal{F}$.  To complete the proof, we show that (i) $\mathcal{D}$  can be converted to an equivalent $\rbt$ with polynomial blow-up, and (ii)  
	$\rbt$s are closed under composition with a polynomial complexity. 
\end{enumerate}

Thus,  our results extend~\cite{DFJL17} to the setting of infinite words, retaining the expressivity of deterministic machines and a polynomial complexity for composition.
The main challenges  when going to the infinite word setting is in dealing with the acceptance conditions. Unlike the finite word setting where acceptance is something to take care of at the end, here we need to deal with it throughout the run. The difficulty in doing this 
comes from the fact that we cannot compute the set of co-accessible states at a given input position. It must be noted that in the proof \cite{DFJL17} for finite words, the equivalent reversible transducer was constructed by computing the set of accessible and co-accessible states at each position of the input word. Indeed,  computing the co-accessible states
at each input position requires an infinite computation or an oracle, and hence, the proof of \cite{DFJL17} fails for infinite words. Instead, we introduce the  
intermediate model of $\cbsst$ where we employ a dedicated ``out'' register that serves as the output tape.  

Our result of extending \cite{DFJL17} can be seen as a positive contribution to reactive synthesis : transforming  specifications over infinite word transformations to an equivalent  $\rbt$  can give rise to efficient solutions to algorithmic problems on transformation specifications. To the best of our knowledge, there is no such translation for infinite words; the    closest result in this direction, but for finite words, is \cite{DBLP:conf/lics/DartoisGGK22}, which gives an efficient procedure for converting specifications  given as RTE (regular transducer expressions) to reversible transducers. 

\smallskip 

Continuing with transformations on infinite words, \cite{DFKL22} investigated  a practical question on functions over infinite words, namely, ``given a function over infinite words, is it computable?''.  They established that the decidability of this question  boils down to checking the continuity of these functions. Further, they conjectured 
that any continuous regular function can be computed by a deterministic two-way transducer  
over infinite words without $\omega$-regular look-ahead. \cite{CD22} took up this conjecture and showed that any continuous rational function over infinite words can be extended to a function which is computable by deterministic two-way transducers  
over infinite words without $\omega$-regular look-ahead. Most recently, \cite{CDFW23} 
conjectured that deterministic two-way transducers with the B\"uchi acceptance condition capture the class of continuous, regular functions. 

\smallskip

Apart from its application to synthesis, $\dbt$ also realize continuous functions. 
This implies that the  conjecture of \cite{CDFW23} fails, since $\dbt$ are more expressive than the class of deterministic two-way transducers with B\"uchi acceptance. A simple example 
illustrating this is the function $f: \{a,b\}^{\omega} \rightarrow \{a,b\}^{\omega}$ such that $f(u)=u$ if the number of $a$'s in $u$ is finite, and is undefined otherwise. $f$ is continuous since it is continuous on its domain; $f$ cannot be realized by a deterministic transducer with B\"uchi acceptance, but it can be realized by a $\dbt$. 
Note however that the extension is only able to refine the domain, and not the production.
In particular, by simply dropping the accepting condition of an $\rbt$, we obtain a function realized by a deterministic two-way transducer with a \buchi condition.
Moreover, our constructions (going from deterministic two-way to reversible) for this class become simpler. 
And conversely, we show that two-way reversible transducers with no acceptance condition  have the same expressiveness as those with the B\"uchi acceptance condition, which in turn have the same expressiveness as two-way deterministic transducers with  B\"uchi acceptance.

\smallskip 

\noindent{\bf{Organization of the Paper}}. Section \ref{sec:prelim} defines the two models we introduce in the paper, namely, $\cbsst$ and $\dbt$. Section \ref{sec:main} states our main result : starting from a $\dbt$, we can obtain an equivalent $\rbt$. Most of the  remaining sections are devoted to the proof of this result. In sections \ref{sec:compose} and \ref{sec:owd2twr} respectively, we prove the closure under composition of $\rbt$s with a polynomial complexity and the polynomial conversion from one-way parity transducers to $\rbt$. Section \ref{sec:cpsst2rpt} uses both these results, where we 
describe the conversion from $\cbsst$ to $\rbt$ with a polynomial complexity. Section 
\ref{sec:dpt2cpsst} contains one of the most non-trivial constructions of the paper, namely, going from $\dbt$ to $\cbsst$ with a single  exponential blow-up. Finally, Section \ref{sec:cont} wraps up by discussing the connection between continuity and the topological closure of $\dbt$s.

\section{Preliminaries}
\label{sec:prelim}
Let $\alp$ be an \emph{alphabet}, i.e., a finite set of letters. 
A finite or infinite word $w$ over $\alp$ is a (possibly empty) sequence
$w=a_{0}a_{1}a_{2}\cdots$ of letters $a_{i}\in\alp$.
The set of all finite (resp.\ infinite) words is denoted by $\alp^*$ (resp.\ 
$\alp^{\omega}$), with $\epsilon$ denoting the empty word. We let 
$\alp^{\infty}=\alp^{*}\cup\alp^{\omega}$.
A \emph{language} is a subset of the set of all words.

\subsection*{Two-way Parity Automata and Transducers}
Let $\alp$ be a finite alphabet and let $\leftend\notin \alp$ be a left delimiter symbol.
We write $\alp_{\leftend}=\alp \cup \{\leftend\}$.

A two-way parity automaton ($\tfa$) is a tuple $\Aa=(Q, \alp, \Delta, q_0,\chi)$, 
where the finite set of states $Q$ is partitioned into a set of forward states $Q^+$ and a
set of backward states $Q^-$.  The initial states is $q_0 \in Q^+$, $\Delta \subseteq Q
\times \alp_{\leftend} \times Q$ is the transition relation and $\chi$ is a finite set of 
coloring functions $\cc\colon\Delta\to\Nat$ which are used to define the acceptance 
condition.
We assume that if $(p, \leftend, q) \in \Delta$, then $p\in Q^-$ and $q \in Q^+$: on
reading $\leftend$, the reading head does not move.

A configuration of a $\tfa$ over an input word $w\in\alp^{\omega}$ is some $\leftend ~u
~p~v$ where $p \in Q$ is the current state and $u \in \alp^*$, $v \in \alp^{\omega}$ with
$w=uv$.
The configuration admits several successor configurations as defined below.
\begin{enumerate}
    \item If $p \in Q^+$, then the input head reads the first symbol $a\in A$ of the 
    suffix $v=av'\in\alp^{\omega}$.
    Let $(p, a, q) \in \Delta$ be a transition.  If $q \in Q^+$, then the successor
    configuration is $\leftend ~ua~q~v'$.  Likewise, if $q \in Q^-$, then the successor
    configuration is $\leftend ~u~q~av'$.  Thus, if the current and target states are both
    in $Q^+$, then the reading head moves right.  If the current state is forward and the
    target state is backward, then the reading head does not move.

    \item If $p \in Q^-$, then the input head reads the last symbol $a\in\alp_{\leftend}$
    of the prefix $\leftend u$.
    Let $(p, a, q) \in \Delta$ be a transition.  If $q \in Q^+$, the successor
    configuration is $\leftend~u~q~v$.  If $q \in Q^-$ then $a\neq\leftend$, we write
    $u=u'a$ with $u'\in\alp^*$ and the successor configuration is $\leftend~u'~q~av$.
    Thus, if the current state is backward and the target state is forward, the reading
    head does not move.  If both states are backward, then the reading head moves left.
\end{enumerate}

A run $\rho$ of $\Aa$ is a finite or infinite sequence of configurations starting from an
initial configuration $\leftend~\varepsilon~q_0~w$ where $w \in \alp^{\omega}$ is the input word:
$$
\leftend q_0w = \leftend u_{0}q_0v_{0} \xrightarrow{} \leftend u_1q_1v_1 \xrightarrow{}
\leftend u_2q_2v_2 \xrightarrow{} \leftend u_3q_3v_3 \xrightarrow{} \leftend u_4q_4v_4 \cdots
$$
We say that $\rho$ reads the whole word $w\in\alp^{\omega}$ if $\supr\{|u_n| \mid
n>0\}=\infty$.  The set of transitions used by $\rho$ infinitely often is denoted
$\infi(\rho)\subseteq \Delta$.
The word $w$ is accepted by $\Aa$, i.e., $w\in\dom{\Aa}$ if $\rho$ reads the whole word
$w$ and $\min(\cc(\infi(\rho)))$ is even for all $c\in\chi$.
Note that, even though we call our machine parity, we in fact consider a conjunction of
parity conditions.  This allows to easily describe intersection of automata or composition
of transducers.

The parity automaton $\Aa$ is called
\begin{itemize}[nosep]
  \item \emph{one-way} if $Q^{-}=\emptyset$,

  \item \emph{deterministic} if for all pairs $(p,a)\in Q\times\alp_{\leftend}$, there is
  at most one state $q=\delta(p,a)$ such that $(p,a,q)\in\Delta$, in this case we
  identify the transition relation $\Delta$ with the partial function $\delta\colon
  Q\times\alp\to Q$,
	
  \item \emph{co-deterministic} if for all pairs $(q,a)\in Q\times\alp_{\leftend}$, there is
  at most one state $p$ such that $(p,a,q)\in\Delta$,
	
  \item \emph{reversible} if it is both deterministic and co-deterministic.
\end{itemize}

A two-way parity transducer $\twt$ is a tuple $\Tt=(Q, \alp, \Delta, q_0, \chi, \alpo,
\lambda)$ where $\Aa=(Q, \alp, \Delta, q_0, \chi)$ is a \emph{deterministic} $\tfa$,
called the underlying parity automaton of $\Tt$, $\alpo$ is a finite \emph{output}
alphabet, and $\lambda\colon\Delta\rightarrow \alpo^*$ is the output function.  As in
the case of $\tfa$, a $\twt$ is one-way/co-deterministic/reversible if so is the
underlying parity automaton.  Let $\dbt$ (resp.\ $\rbt$) denote two-way (deterministic) 
(resp.\ reversible) parity transducers.  The notion of run and accepting run is
inherited from the underlying $\tfa$.  For $w \in \alp^{\omega}$ such that $w \in
\dom{\Aa}$, let the accepting run $\rho$ of $w$ be
$$
\leftend q_0w = \leftend u_{0}q_0v_{0} \xrightarrow{t_{1}} \leftend u_1q_1v_1 
\xrightarrow{t_{2}} \leftend u_2q_2v_2 \xrightarrow{t_{3}} \leftend u_3q_3v_3 
\xrightarrow{t_{4}} \leftend u_4q_4v_4 \cdots
$$
where $t_{i}\in\Delta$ is the $i$-th transition taken during the run, i.e., from 
$\leftend u_{i-1}q_{i-1}v_{i-1}$ to $\leftend u_{i}q_{i}v_{i}$.
For $i>0$, let $\gamma_{i}=\lambda(t_i)$ be the output produced by the $i$-th
transition of $\rho$.  If $\gamma_1\gamma_2\gamma_3\gamma_4\cdots\in\alpo^{\omega}$, then
$w\in\dom{\Tt}$ and we let $\sem{\Tt}(w)=\gamma_1\gamma_2\gamma_3\gamma_4\cdots$ be the
output word computed by $\Tt$.  Hence, the semantics of a $\twt$ is a partial function
$\sem{\Tt}\colon\alp^{\omega}\to\alpo^{\omega}$ with $\dom{\Tt}\subseteq\dom{\Aa}$.

 \begin{figure}
	
	\begin{tikzpicture}[scale=0.8]
		
		\node[state, initial, initial text=, scale=0.7] (p) at (0,0) {$+$};
		\node[state, scale=0.7] (q) at (2.5,0) {$-$};
		\node[state, scale=0.7] (r) at (5,0) {$+$};
		
		\path[->] (p) edge  [loop above] node {$a|a : 0$} (p)
		(q) edge  [loop above] node {$a|a : 1$} (q)
		(r) edge  [loop above] node {$a|\epsilon : 1$} (r)
		(p) edge [bend left=7,above] node {\small $\#|\#$ : 1} (q)
		(q) edge [bend left=7,above] node {\small $\#,\vdash|\epsilon : 1$} (r)
		(r) edge [bend left,below] node {\small $\#|\# : 0$} (p);
		
		\node[initial,initial text=,initial below,state, scale=0.7] (sst) at (11,0) {};
		
		\path[->] (sst) edge  [loop left] node {\small $a\begin{cases}\outreg=\outreg\ a \\
				X=aX \end{cases}$} (sst)
		(sst) edge  [loop right] node {\small $\#\begin{cases}\outreg=\outreg\#X\#\\
				X=\epsilon\end{cases}$} (sst);

	\end{tikzpicture} 
 \caption{An example of $\rbt$ (left) and $\cbsst$ (right) defining the function \emph{map-copy-reverse} ($mcr$) defined on $(A\uplus \{\#\})^\omega \to (A\uplus \{\#\})^\omega$  by:
   $mcr(u_1\#u_2\#...)=u_1\#\tilde{u_1}\#u_2\#\tilde{u_2}\#...$ for words with an infinite number of letter $\#$, and   
$mcr(u_1\#\ldots\#u_n\#u)=u_1\#\tilde{u_1}\#\ldots u_n\#\tilde{u_n}\#u$ if $u\in A^\omega$, where $\tilde{v}$ denotes the mirror image of $v$.
There is only one coloring function, denoted on the transitions after the colon. The color of all transitions of the \( \cbsst \) is \( 0 \).}
\end{figure}
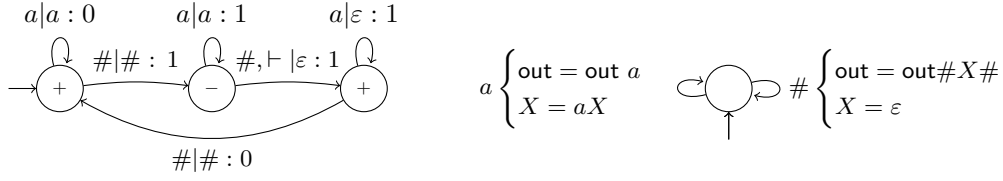

\subsection*{Parity Streaming String Transducers}

Let $\R$ be a finite set of variables called \emph{registers}.
A \emph{substitution} of $\R$ into an alphabet $\alpo$ is a mapping  
$\sigma\colon\R\to(\R\uplus\alpo)^*$.
It is called \emph{copyless} if for all $r\in\R$, $r$ appears at most once in the
concatenation of all the $\sigma(r')$ for $r'\in\R$.
We denote by  $\subst{\alpo}{\R}$ the set of all copyless substitutions of $\R$ into $\alpo$.

A copyless parity Streaming String Transducers ($\cbsst$) is given by a tuple
$\Tt=(Q,\alp,\Delta,q_{0},\chi,\alpo,\R,\outreg,\lambda)$ where
$\Aa=(Q,\alp,\Delta,q_{0},\chi)$ is a deterministic one-way parity automaton called the
underlying parity automaton of $\Tt$,
$\R$ is a finite set of registers,
$\outreg\in \R$ is a distinguished register, called the output register,
$\lambda\colon\Delta %
\to\subst{\alpo}{\R}$ is the update function satisfying additionally 
$\lambda(t)(\outreg)\in\outreg\cdot(\R\uplus\alpo)^*$ for all $t\in\Delta$.
  
A configuration of a copyless parity SST $\Tt$ is a tuple $(q,\nu)$ where $q \in Q$ and
$\nu\colon\R\to\alpo^*$ is an assignment.
The initial configuration is $(q_{0},\nu_{0})$ where $\nu_{0}(r)=\varepsilon$ for all $r\in R$.
Since the automaton $\Aa$ is \emph{deterministic}, we simply describe a run on an input 
word $w=a_{0}a_{1}a_{2}\cdots$ as a sequence a sequence of transitions applying the 
corresponding substitutions to the assignments:
$$
(q_{0},\nu_{0})\xrightarrow{a_{0}}(q_{1},\nu_{1})\xrightarrow{a_{1}}(q_{2},\nu_{2})
\xrightarrow{a_{2}}(q_{3},\nu_{3})\cdots
$$ 
where $(q_{0},\nu_{0})$ is the initial configuration and for all $i\geq0$ we have
$t_{i}=(q_{i},a_{i},q_{i+1})\in\Delta$ and 
$\nu_{i+1}=\nu_{i}\circ\lambda(t_{i})$\footnote{An assignment $\nu\colon\R\to\alpo^{*}$ 
is extended to a morphism $\nu\colon(\R\uplus\alpo)^{*}\to\alpo^{*}$ by $\nu(b)=b$ for 
all $b\in\alpo$. Hence, if $\sigma\in\subst{\alpo}{\R}$ is a substitution then 
$\nu'=\nu\circ\sigma$ is an assignment defined by $\nu'(r)=\nu(\sigma(r))$ for all $r\in\R$.
For instance, if $\sigma(r)=br'cbr$ then $\nu'(r)=b\nu(r')cb\nu(r)$.}.
Notice that, from the restriction of the update function, we deduce that 
$\nu_{0}(\outreg),\nu_{1}(\outreg),\nu_{2}(\outreg),\ldots$ is a (weakly) increasing 
sequence of output words in $\alpo^{*}$. If this sequence is unbounded then 
$w\in\dom{\Tt}$ and we let
$\sem{\Tt}(w)=\bigsqcup_{i\geq0}\nu_{i}(\outreg)\in\alpo^\omega$ be the limit (least 
upper-bound) of this sequence. Hence, the semantics of a $\cbsst$ is a partial function
$\sem{\Tt}\colon\alp^{\omega}\to\alpo^{\omega}$ with $\dom{\Tt}\subseteq\dom{\Aa}$.

\section{Main Result}
\label{sec:main}
We are now ready to state our main result, which is an effective procedure to construct a reversible two-way transducer for a deterministic machine. Our result is stated using a conjunction of parity conditions.

The proof relies on constructions that go through $\cbsst$, and are presented in the subsequent sections.

\begin{theorem}\label{thm-main}
Given a deterministic $\dbt$ $T$ with $n$ states, $k$ color conditions and $\ell$ colors, we can construct a $\rbt$ $S$ with $O(\ell^{2kn}(2n)^{4n+1}))$ states, $k$ color conditions and $\ell$ colors such that $\sem{T}=\sem{S}$. 
\end{theorem}

\begin{proof}
Let $T$ be a  $\dbt$ with $n$ states, $k$ color conditions and $\ell$ colors.
Using Theorem~\ref{thm:2DTtoSST}, we can construct an equivalent $\cbsst$ $T'$ with $O(n(\ell^k)^n(2n+1)^{2n-1})$ states, $2n$ variables, $k$ color conditions and $\ell$ colors. 
Then by Theorem~\ref{thm-SSTtoR2W}, we can construct a $\rbt$ $S$ equivalent to $T'$ whose size is quadratic in the number of states and linear in the number of variables.
More precisely, $S$ has $O((n(\ell^k)^n(2n+1)^{2n-1})^2(2n))=O(\ell^{2kn}(2n)^{4n+1}))$ states, $k$ color conditions and $\ell$ colors, concluding the proof.
\end{proof}

\section{Composition of $\rbt$}\label{sec:compose}
The main reason to use reversible two-way machines is that they are easily composable. Given two composable reversible transducers, we can construct one whose size is linear in both machines, and whose transition function is rather straight-forward. It is explicited in the following theorem and proof.

\begin{theorem}\label{thm:composition}
  Given two $\rbt$ $\Ss$ and $\Tt$, of size $n$ and $m$ respectively, and such that the
  output alphabet of $\Ss$ is the input alphabet of $\Tt$, we can construct a $\rbt$
  $\Uu$, also denoted by $\Tt\circ\Ss$, of size $O(nm)$ such that
  $\sem{\Uu}=\sem{\Tt}\circ\sem{\Ss}$.
\end{theorem}

\begin{proof}[Sketch of proof.]
  The set of states of the machine $\Uu$ is the cartesian product of the sets of states of
  $\Ss$ and $\Tt$.  Given an input word $u$ of $\Ss$, $\Uu$ simulates $\Ss$ until some
  transition produces a nonempty output word $v\in\alpo^{+}$.  Then, it stops the
  simulation of $\Ss$ to simulate the run $\rho_\Tt$ of $\Tt$ over $v$.  If $\rho_\Tt$
  exits $v$ on the right, then $\Uu$ resumes the simulation of $\Ss$ up to the next
  transition producing a nonempty word.  Otherwise, it rewinds the run of $\Ss$ to get its
  previous production, and simulates $\Tt$ on it, starting from the right.

  The conjunction of parity conditions allows to an easy construction for intersection,
  which is similar to what is expected here.  A word $u$ should be accepted if $u$ belongs
  to the domain of $\Ss$ and $\sem{\Ss}(u)$ belongs to the domain of $\Tt$.  By doing the
  conjunction of both acceptance, we are able to recognize the domain of $\Uu=\Tt\circ\Ss$.
\end{proof}

\begin{proof}
  Let $\Ss=(Q, \alp, \delta, q_0, \chi, \alpo, \lambda)$ and 
  $\Tt=(P, \alpo, \alpha, p_0, \chi', C, \beta)$.
  We define the composition $\Uu=\Tt\circ\Ss=(R,\alp,\mu,r_0,\chi'',C,\nu)$ where 
  $R=Q\times P$ is splitted as
  $$
  R^+=Q^+\times P^+ \cup Q^-\times P^- \qquad\qquad\qquad
  R^-=Q^-\times P^+ \cup Q^+\times P^- \,.
  $$
  The initial state is $r_0=(q_0,p_0)$ and $\mu$, $\nu$ and $\chi''$ are defined below.
  
  To properly define $\mu$ and $\nu$, we extend $\alpha$ and $\beta$ to finite words, and
  more precisely to the productions of $\Ss$.  Given a word $v=\lambda(q,a)\in\alpo^{*}$
  for some $(q,a)\in Q\times \alp$, and a state $p$ of $\Tt$, we define $\rho_p(v)$ to be
  the maximal run of $\Tt$ over $v$ starting in state $p$ on the left (resp.\ right) of
  $v$ if $p\in P^+$ (resp.\ $p\in P^-$).
  Then we define $\alpha^*(p,v)$ as the state reached by
  $\rho_{p}(v)$ when exiting $v$.  It is undefined if $\rho_{p}(v)$ loops within $v$.
  Note that if $\alpha^*(p,v)$ belongs to $P^+$ (resp.\ $P^-$), then $\Tt$ exits $v$ on
  the right (resp.\ on the left).  We define $\beta^*(p,v)$ as the concatenation of the
  productions of $\rho_p(v)$.  If $\alpha^*(p,v)$ is defined, then $\beta^*(p,v)$ is
  finite.  Note that $\rho_{p}(\varepsilon)$ is an empty run, so we have
  $\alpha^*(p,\epsilon)=p$ and $\beta^*(p,\epsilon)=\epsilon$.

  For the parity conditions, we let $\chi''=\{\overline{\cc}\mid\cc\in\chi\cup\chi'\}$ 
  and we extend the functions $\cc\in\chi'$ to finite runs $\rho_{p}(v)$. More precisely, 
  we let $\cc^{*}(p,v)$ be the minimum $\cc$-value taken by the transitions of 
  $\rho_{p}(v)$. When $v=\varepsilon$ then $\rho_{p}(v)$ is an empty run and we 
  set $\cc^{*}(p,v)$ to the largest odd value in all values taken by $\cc$ on transitions of 
  $\Tt$.
  Then, given a state $(q,p)$ of $\Uu$,
  \begin{itemize}
    \item If $p\in P^+$ then we let $v=\lambda(q,a)$.  We set $\nu((q,p),a)=\beta^*(p,v)$
    and, %
    with $q'=\delta(q,a)$ and $p'=\alpha^*(p,v)$ we define
    $$
    \mu((q,p),a)=
    \begin{cases}
      (q',p') & \text{if } p'\in P^+, \\
      (q ,p') & \text{if } p'\in P^-
    \end{cases}
    \qquad\text{and}\qquad
    \overline{\cc}((q,p),a)=
    \begin{cases}
      \cc(q,a) & \text{if } \cc\in\chi, \\
      \cc^{*}(p,v) & \text{if } \cc\in\chi'.
    \end{cases}
    $$

    \item If $p\in P^-$ then we let $q'$ be such that $q=\delta(q',a)$ and
    $v=\lambda(q',a)$.  Note that $q'$ is unique by co-determinism of $\Ss$.  We set
    $\nu((q,p),a)=\beta^*(p,v)$ and, with $p'=\alpha^*(p,v)$ we define
    $$
    \mu((q,p),a)=
    \begin{cases}
      (q ,p') & \text{if } p'\in P^+, \\
      (q',p') & \text{if } p'\in P^-. 
    \end{cases}
    \qquad\text{and}\qquad
    \overline{\cc}((q,p),a)=
    \begin{cases}
      \cc(q',a) & \text{if } \cc\in\chi, \\
      \cc^{*}(p,v) & \text{if } \cc\in\chi'.
    \end{cases}
    $$    

  \end{itemize}
  The intuition behind the transition function is that $\Uu$ simulates $\Ss$ to feed a
  simulation of $\Tt$.  If $\Tt$ moves forward on its input, then $\Uu$ simulates $\Ss$ forward.
  If $\Tt$ moves backward on its input, then $\Uu$ backtracks the computation of $\Ss$.
  During a switch of direction, $\Ss$ stays put.
  
  The acceptance condition of conjunctive parity was chosen specifically to allow for
  smooth composition.  By using both sets of parities, the transducer $\Uu$ ensures that
  the input word is accepted by $\Ss$, and that its production is accepted by $\Tt$.  It
  is worth noting that since $\Uu$ can rewind the run of $\Ss$, it can take a transition
  (and consequently its colors) multiple times on a given input position.
  This increases the multiplicity of the transitions taken during a non-looping run by a
  constant factor since a deterministic transducer never visits twice a given position in
  the same state.  Hence, the set of colors of $\Ss$ that $\Uu$ sees infinitely often on
  a non-looping run is the same as the ones seen by $\Ss$.
\end{proof}

\section{\( \odbt \) to \( \rbt \)}\label{sec:owd2twr}

Similarly to the finite words, given a deterministic one-way machine, one can construct a reversible one realizing the same function.

\begin{theorem}\label{thm:odbtEqRdbt}
	Let \( \Tt \) be a \( \odbt \) with \( n \) states, we can construct a \( \rbt \) \( \Tt' \)
	of size \( O(n^2) \)
	such that \( \sem{\Tt} = \sem{\Tt'} \).
\end{theorem}

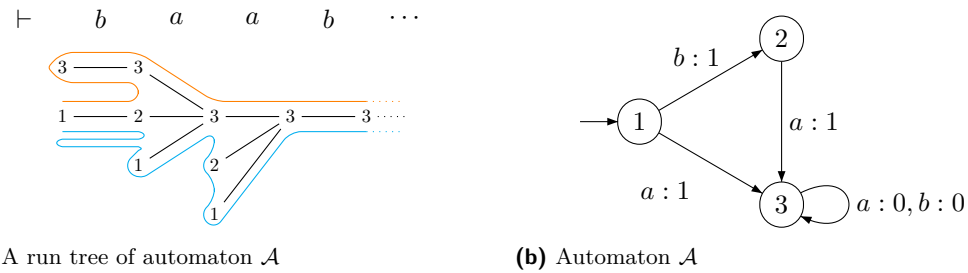
\begin{figure}[tbp]
	\centering
	\begin{subfigure}{0.48\textwidth}
		\def\xspacing {1cm}
		\def\yspacing {0.65cm}
		\def\smallspacing {0.2cm}
		\def\verysmallspacing {0.1cm}
		\newcommand{\abvcoord}[1]{($(#1)-(0,\smallspacing)$)}
		\newcommand{\littleabvcoord}[1]{($(#1)-(0,\verysmallspacing)$)}
		\newcommand{\blwcoord}[1]{($(#1)-(0,-\smallspacing)$)}
		\newcommand{\leftcoord}[1]{($(#1)+(-\smallspacing,0)$)}
		\centering
		\begin{tikzpicture}
			\coordinate (p0) at (0, 0);
			
			\foreach \index in {1, 2, 3, 4}{
				\def\indexprev{\the\numexpr\index-1\relax}
				\coordinate (p\index) at ($(p\indexprev)+(\xspacing, 0)$);
			}
			
			\node (l0) at ($(p0)-(\xspacing/2, 0)$) {\( \leftend \)};
			
			\foreach \index/\letter in {1/b, 2/a, 3/a, 4/b}{
				\node (l\index) at ($(p\index)-(\xspacing/2, 0)$) {\( \letter \)};
			}
			
			\node (ldots) at ($(p4)+(\xspacing/2, 0)$) {\( \cdots \)};
			
			\coordinate (baserun) at (-\xspacing, -2*\yspacing);
			
			\foreach \index/\stateid in {1/1, 2/2, 3/3, 4/3, 5/3}{
				\node[scale=0.7] (br\index) at ($(baserun)+(\index*\xspacing,0)$) {\( \stateid \)};
			}
			
			\node[coordinate] (brlast) at ($(br5)+(\xspacing/2, 0)$) {};
			
			\foreach \index in {2, 3, 4, 5}{
				\def\indexprev{\the\numexpr\index-1\relax}
				\draw (br\indexprev) -- (br\index);
			}
			
			\draw[dotted] (br5) -- (brlast);
			
			\coordinate (belowrun) at ($(baserun)+(\xspacing, \yspacing)$);

			\foreach \index/\stateid in {0/3, 1/3}{
				\node[scale=0.7] (bel\index) at ($(belowrun)+(\index*\xspacing,0)$) {\( \stateid \)};
			}
			
			\draw (bel0) -- (bel1);
			\draw (bel1) -- (br3);

			\node[scale=0.7] (above1) at ($(p1)-(0,\yspacing*3)$) {1};
			\node[scale=0.7] (above22) at ($(p2)-(0,\yspacing*3)$) {2};
			\node[scale=0.7] (above21) at ($(p2)-(0,\yspacing*4)$) {1};
			
			\foreach \first/\second in {above1/br3, above21/br4, above22/br4}{
				\draw (\first) -- (\second);
			}
			
			\draw[rounded corners, color=orange] \blwcoord{br1} -- \blwcoord{br2} -- \abvcoord{bel1} -- \abvcoord{bel0} .. controls \leftcoord{bel0} .. \blwcoord{bel0} -- \blwcoord{bel1} -- \blwcoord{br3} -- \blwcoord{br4} -- \blwcoord{br5};
			
			\draw[dotted, color=orange] \blwcoord{br5} -- \blwcoord{brlast};
			
			\node[coordinate] (abvbr1) at \abvcoord{br1} {};
			\node[coordinate] (abvbr2) at \abvcoord{br2} {};
			\node[coordinate] (abvabvbr1) at \littleabvcoord{abvbr1} {};
			\node[coordinate] (abvabvbr2) at \littleabvcoord{abvbr2} {};
			\node[coordinate] (topbr1) at \littleabvcoord{abvabvbr1} {};
			\node[coordinate] (topbr2) at \littleabvcoord{abvabvbr2} {};
			
			\node[coordinate] (controlbr21) at ({$(abvbr2)!0.2!(abvabvbr2)$} -| {$(abvbr2)+(\verysmallspacing, 0)$}) {};
			\node[coordinate] (controlbr22) at ({$(abvbr2)!0.8!(abvabvbr2)$} -| {$(abvbr2)+(\verysmallspacing, 0)$}) {};
			
			\node[coordinate] (controlbr11) at ({$(abvabvbr1)!0.2!(topbr1)$} -| {$(abvabvbr1)+(-\verysmallspacing, 0)$}) {};
			\node[coordinate] (controlbr12) at ({$(abvabvbr1)!0.8!(topbr1)$} -| {$(abvabvbr1)+(-\verysmallspacing, 0)$}) {};
			
			\node[coordinate] (controltop2) at ({$(topbr2)!0.5!\blwcoord{above1}$} -| {$(topbr2)+(\verysmallspacing/2, 0)$}) {};
			
			\draw[rounded corners, color=cyan] (abvbr1) [sharp corners] -- (abvbr2) .. controls (controlbr21) and (controlbr22) .. (abvabvbr2) -- (abvabvbr1) .. controls (controlbr11) and (controlbr12) .. (topbr1) -- (topbr2) .. controls (controltop2) .. \blwcoord{above1} [rounded corners]  -- \leftcoord{above1} -- \abvcoord{above1} -- \abvcoord{br3} -- \blwcoord{above22} -- \leftcoord{above22} -- \abvcoord{above22} -- \blwcoord{above21} -- \leftcoord{above21} -- \abvcoord{above21} -- \abvcoord{br4} -- \abvcoord{br5};
			
			\draw[dotted, color=cyan] \abvcoord{br5} -- \abvcoord{brlast};
		\end{tikzpicture}
		\caption{A run tree of automaton \( \Aa \)}
		\label{fig:ewd2twr:runtree}
	\end{subfigure}
	\hfill
	\begin{subfigure}{0.48\textwidth}
		\centering
		\gusepicture{automaton}
		\caption{Automaton \( \Aa \)}
		\label{fig:ewd2twr:a}
	\end{subfigure}
	\caption{The automaton \( \Aa \) depicted in Figure~\ref{fig:ewd2twr:a}
		recognizes all infinite words over alphabet \( \left\{ a, b \right\} \)
		having an \( a \) in first or second position
		(it has only one coloring function,
		represented after the colons in the transitions).
		Figure~\ref{fig:ewd2twr:runtree} is the part of the tree
		corresponding to the run of \( \Aa \) on the prefix \( baab \)
		of an accepted word.
		Each node of the tree is a configuration of \( \Aa \),
		represented here by a control state.
		Its horizontal position allows to deduce the position
		in the input word, depicted above the tree.
		The horizontal straight path represents the accepting run.
		Notice that when the top reading head needs to go backward
		to go around a branch,
		the bottom ones follows and goes backward on the accepting run.}
	\label{fig:ewd2twr}
\end{figure}

\begin{proof}
	The construction is reminiscent of the tree-outline construction for co-deterministic
	transducers of \cite{DFJL17}.  The difference is that here, we begin with a
	deterministic transducer with an infinite input word, so instead of starting from the
	root of the tree, which is at the end of a finite input word, our outline has to start
	from a leaf at the beginning of the input word, corresponding to the initial
	configuration.  We also generalize the construction by allowing any degree of non
	(co-)determinism: while in \cite{DFJL17}, at most two branches could merge on any
	vertex, here we allow any number.
	
	We begin with the underlying automaton:
	from a one-way deterministic parity automaton (\( \odba \)) \( \Aa \),
	we build a two-way reversible parity automaton (\( \rba \)) \( \Aa' \)
	simulating the behavior of \( \Aa \).
	For any accepted input word \( w \),
	we consider the infinite acyclic graph (simply called a tree) representing all the 
	partial runs of \( \Aa \) merging with the accepting run of $\Aa$ on \( w \)
	(note that because \( \Aa \) is deterministic,
	there is only one accepting run for a given word).
	Automaton \( \Aa' \) will simulate two synchronized reading heads
	going along the outline of this tree,
	as illustrated in Figure~\ref{fig:ewd2twr}.
	
	The two heads are required to make \( \Aa' \) equivalent
	to \( \Aa \): we need to be able to discriminate configurations of
	a run of \( \Aa' \) occurring in the accepting run of \( \Aa \)
	from the ones added to account for the non-initial runs.
	One reading head follows the outline of the run tree from above,
	and the other one from below.
	The configurations where the two reading heads point to the same
	state of \( \Aa \) correspond to those occurring in the accepting run of
	this automaton.
	
	The reading heads are placed above and below the initial state,
	and they move together to the right,
	until one of them encounters a branching in the tree.
	When this happens, 
	the $\rba$ moves backwards to go around the branch.
	When the branch dies
	(which necessarily happens because from each position, the prefix of a word is finite),
	the exploration continues to the right.
	As the two heads are synchronized, and because branches may not all be of the same length,
	when one head needs going left the other one may impose right moves
	in order to reach another branch,
	on which it will be able to go left far enough to follow the first head.

	A run of \( \Aa' \) can be seen as a straightforward journey
	along the flattened outline of the tree,
	hence the reversibility.

	Once \( \Aa' \) is defined, we define \( \Tt' \):
	it is the \( \rbt \) having \( \Aa' \) as underlying automaton,
	and whose output function is that of \( \Tt \) from states where the two reading heads
	point to the same state, and \( \varepsilon \) otherwise.

	\noindent\textbf{Formal construction of \( \Aa' \).}
	Let \( \Tt= \left( Q, \alp, \delta, q_0, \chi, \alpo, \lambda \right) \) be a $\odbt$.
	
	Fix an arbitrary total order $\preceq$ over \( Q \).
	For $q\in Q$ and $a\in\alp$ we let $\delta_{a}^{-1}(q)=\{p\in Q\mid\delta(p,a)=q\}$, 
	and we also set $\delta_{\leftend}^{-1}(q)=\emptyset$,
	except if \( q \neq q_0 \) (it is undefined otherwise).
	Let \( \predabv{a}{q}{q'} \) be a predicate,
	true if \( q' \) is minimal with respect to \( \prec \) such that
	\( q \prec q' \) and $\delta(q,a)=\delta(q',a)$.
	Let $\overline{Q}=\{\overline{q}\mid q\in Q\}$ and
	$\underline{Q}=\{\underline{q}\mid q\in Q\}$
	be two copies of $Q$.
	Define \( \Aa' = \left( Q', \alp, \delta, q_0', \chi' \right) \) by
	\begin{itemize}[nosep]
		\item \( Q' = Q'^{+} \uplus Q'^{-} =
		((\underline{Q}\cap\overline{Q})\times(\underline{Q}\cap\overline{Q})) \setminus
		\{ (\underline{q},\underline{q}),(\overline{q},\overline{q}) \mid q\in Q \}\)
		with
		\( q_0'=(\underline{q_0},\overline{q_0}) \) and
		\begin{itemize}[nosep]
			\item \( Q'^{+} = \underline{Q} \times \overline{Q} \cup
			\overline{Q} \times \underline{Q} \), 
			\item \( Q'^{-} =  \left( \underline{Q} \times \underline{Q} \cup
			\overline{Q} \times \overline{Q} \right) \setminus
			\left\{ \left( \overline{q}, \overline{q} \right),
			\left( \underline{q}, \underline{q} \right) \mid
			q \in Q \right\} \).
		\end{itemize}
		In a state $(r,s)\in Q'$, the first (resp.\ second) component
		is for the head which is
		\textquote{above} (orange line) (resp.\ \textquote{below} (blue line))
		the accepting run (black straight line) in Figure~\ref{fig:ewd2twr:runtree}.
		In both cases, a state $\underline{q}\in\underline{Q}$
		(resp.\ $\overline{q}\in\overline{Q}$) means that the
		corresponding head (colored line) is above (resp.\ below) the state.
		
		\item Transitions: first two cases for $Q'^{+}$ states and then for
		$Q'^{-}$ states
		\begin{enumerate}[nosep]
			\item  \label{item:1}
			$\delta'((\underline{p},\overline{q}),a)=
			\begin{cases}
				(\overline{p'},\overline{q}) & \text{if $\predabv{a}{p}{p'}$ for some $p'\in Q$} \\
				(\underline{p},\underline{q'}) & \text{elseif $\predabv{a}{q'}{q}$ for some $q'\in Q$} \\
				(\underline{\delta(p,a)},\overline{\delta(q,a)}) & \text{otherwise.}
			\end{cases}$
			
			\item  \label{item:2}
			$\delta'((\overline{p},\underline{q}),a)=
			\begin{cases}
				(\underline{p'},\underline{q}) & \text{if $\predabv{a}{p'}{p}$ for some $p'\in Q$} \\
				(\overline{p},\overline{q'}) & \text{elseif $\predabv{a}{q}{q'}$ for some $q'\in Q$} \\
				(\overline{\delta(p,a)},\underline{\delta(q,a)}) & \text{otherwise.}
			\end{cases}$
			
			\item  \label{item:3}
			$\delta'((\overline{p},\overline{q}),a)=
			\begin{cases}
				(\underline{p},\overline{q}) & \text{if $\delta_{a}^{-1}(p)=\emptyset$} \\
				(\overline{p},\underline{q}) & \text{elseif $\delta_{a}^{-1}(q)=\emptyset$} \\
				(\overline{\min\delta_{a}^{-1}(p)},\overline{\min\delta_{a}^{-1}(q)}) & \text{otherwise.}
			\end{cases}$
			
			\item  \label{item:4}
			$\delta'((\underline{p},\underline{q}),a)=
			\begin{cases}
				(\overline{p},\underline{q}) & \text{if $\delta_{a}^{-1}(p)=\emptyset$} \\
				(\underline{p},\overline{q}) & \text{elseif $\delta_{a}^{-1}(q)=\emptyset$} \\
				(\underline{\max\delta_{a}^{-1}(p)},\underline{\max\delta_{a}^{-1}(q)}) & \text{otherwise.}
			\end{cases}$
		\end{enumerate}
		\item $\chi'=\{\cc'\mid\cc\in\chi\}$ with $c'((r,s),a) = 
		\begin{cases}
			c(q,a) &\text{ if } (r,s) = (\underline{q},\overline{q}) \\
			\max \{c(p,a) \mid p\in Q, a\in\alp \} &\text{otherwise.}
		\end{cases}$
	\end{itemize}
	
	\noindent\textbf{Reversibility of \( \Aa' \).}
	From the definition of $\delta'$, $\Aa'$ is clearly deterministic.
	
	We show that \( \Aa' \) is also codeterministic.
	There are three potential transitions leading to a given state in $Q'$ by reading
	a given letter, only one of which can be part of \( \delta' \).
	
	The following case analysis shows this:
	
	\[ \delta_{a}'^{-1}((\underline{p},\overline{q})) =
	\begin{cases}
		\left( \overline{p}, \overline{q} \right) & \text{if } \delta_{a}^{-1}(p) = \emptyset \\
		\left( \underline{p}, \underline{q} \right) & \text{elseif } \delta_{a}^{-1}(q) =
		\emptyset \\
		\left( \underline{\max\delta^{-1}_a(p)}, \overline{\min\delta^{-1}_a(q)} \right) & \text{otherwise} \\
	\end{cases} \]
	
	\[ \delta_{a}'^{-1}((\overline{p}, \underline{q})) =
	\begin{cases}
		\left( \underline{p}, \underline{q} \right) & \text{if } \delta_{a}^{-1}(p) =
		\emptyset \text{ (so \( p \neq q_0 \))} \\
		\left( \overline{p}, \overline{q} \right) & \text{elseif } \delta_{a}^{-1}(q) = \emptyset \\
		\left( \overline{\min\delta^{-1}_a(p')}, \underline{\max\delta^{-1}_a(q')} \right) & \text{otherwise} \\
	\end{cases} \]
	
	\[  \delta_{a}'^{-1}((\overline{p'},\overline{q'}))=
	\begin{cases}
		(\underline{p},\overline{q'}) & \text{if $\predabv{a}{p}{p'}$ for some $p\in Q$} \\
		(\overline{p'},\underline{q}) & \text{elseif $\predabv{a}{q}{q'}$ for some $q\in Q$} \\
		(\overline{\delta(p', a)},\overline{\delta(q', a)}) & \text{otherwise.}
	\end{cases} \]
	
	\[ \delta_{a}'^{-1}((\underline{p'},\underline{q'})) =
	\begin{cases}
		\left( \overline{p}, \underline{q'} \right) & \text{if } \predabv{a}{p'}{p}
		\text{ for some } p \in Q \\
		\left( \underline{p'}, \overline{q} \right) & \text{elseif } \predabv{a}{q'}{q} \\
		\left( \underline{\delta(p', a)}, \underline{\delta(q', a)} \right) & \text{otherwise.}
	\end{cases} \]
	
	We conclude that $\Aa'$ is codeterministic, and therefore reversible.
	
	Intuitively, automaton \( \Aa' \) will follow the run of \( \Aa \),
	adding extra steps to deal with the states that are co-reachable
	from states of this run.
	States of \( \Aa' \) of the form \( \left( \underline{q}, \overline{q} \right) \)
	correspond to states \( q \) in the run of \( \Aa \).
	The key idea behind the construction of \( \Aa' \)
	is that in a run of this automaton,
	configurations of the form \( \sconf{\left( \underline{q}, \overline{q} \right)} \)
	will occur in the same order as the configurations \( \sconf{q} \)
	in the run of \( \Aa \).
	This is the point of the following claim.
	
	\newcommand{\infword}{w}
	\newcommand{\preword}{w'}
	\newcommand{\runaprime}{\rho'}
	
	\begin{claim}\label{claim:1dt2R}
		Let \( \rho = \aconf{q_0}{u_0}{v_0} \xrightarrow{} \aconf{q_1}{u_1}{v_1}
		\xrightarrow{} \aconf{q_2}{u_2}{v_2} \xrightarrow{} \cdots \)
		be an accepting run of \( \Aa \) on \( \infword\in\alp^{\omega} \)
		(we have $w=u_{i}v_{i}$ for all $i\geq0$ where $u_{i}$
		is the prefix of length $i$ of $w$).
		There is an accepting run \( \runaprime \) of \( \Aa' \) on
		\( \infword \) such that the projection of \( \runaprime \)
		on the configurations with states of the form
		\( \left( \underline{p}, \overline{p} \right) \) is
		\( \aninitconf{( \underline{q_0}, \overline{q_0} )}{\infword} \xrightarrow{+}
		\aconf{( \underline{q_1}, \overline{q_1} )}{u_1}{v_1} \xrightarrow{+}
		\aconf{( \underline{q_2}, \overline{q_2} )}{u_2}{v_2} \xrightarrow{+} \cdots .\)
	\end{claim}
	
	\newcommand{\thetree}{T}
	\newcommand{\avertex}{\mathsf{v}}
	\newcommand{\thepath}{\mathsf{p}}

	\begin{proof}[Proof of claim]
		Let $G$ be the configuration graph of $\Aa$ on the input word $w\in\alp^\omega$.
		The vertices of $G$ are all configurations $\aconf{q}{u}v$ with $w=uv$,
		$u\in\alp^{*}$ and $q\in Q$.
		Edges correspond to transitions: we have an edge
		$\aconf{p}{u}{av}\to\aconf{q}{ua}v$ if $\delta(p,a)=q$.
		Since $\Aa$ is deterministic, there is at most one outgoing edge from each 
		configuration and since $\Aa$ is one-way, thus the graph $G$ is acyclic.
		
		Let \( \thetree \) be the connected component of \( G \) that contains the 
		initial configuration \( \initconf{\infword} \).
		Note that the run $\rho$ corresponds to the only infinite path in $G$
		starting from \( \initconf{\infword} \).
		Any configuration $\aconf{p_{i}}{u_{i}}{v_{i}}$ of $T$ which is not on $\rho$ 
		($p_{i}\neq q_{i}$) will eventually merge with $\rho$: 
		$\aconf{p_{i}}{u_{i}}{v_{i}}\xrightarrow{*}\aconf{p_{j-1}}{u_{j-1}}{v_{j-1}}
		\to\aconf{q_{j}}{u_{j}}{v_{j}}$ with $i<j$ and $p_{j-1}\neq q_{j-1}$.
		We say that $\aconf{p_{i}}{u_{i}}{v_{i}}$ is below (resp.\ above) $\rho$ if
		$p_{j-1}\prec q_{j-1}$ (resp.\ $q_{j-1}\prec p_{j-1}$).
		
		Let us consider the run $\runaprime$ of $\Aa'$ on $\infword$. 
		Due to the definition of the transition function of $\Aa'$,
		the run $\runaprime$ only moves along $T$. 
		Indeed a configuration \( \aconf{(r,s)}{u}{v} \) of \( \runaprime \)
		encodes the position of two tokens,
		each placed either above or below a configuration of \( \thetree \).
		Moreover, the first token is always above the branch $\rho$
		while the second token is always below.
		This can be shown by case analysis of \( \delta' \).
		The two transitions where the upper token goes from above a branch to below are
		the following:
		
		\begin{itemize}
			\item \( \delta'\left( \left( \underline{p}, \overline{q} \right), a \right) =
			\left( \left( \overline{p'}, \overline{q} \right) \right) \) if
			\( \left( \predabv{a}{p}{p'} \right) \) :
			there, we know that \( p \prec p' \),
			so the token ends up on a branch that is above where it was;
			
			\item \( \delta'\left( \left( \underline{p}, \underline{q} \right), a \right) =
			\left( \overline{p}, \overline{q} \right)\) if
			\( \delta_{a}^{-1}(p) = \emptyset \) and \( a \neq \leftend \),
			so \( p \) must have reached the end of the branch it was placed on,
			and we know it was not placed on \( \rho \), because when this branch ends,
			\( a = \leftend \).
		\end{itemize}
		
		A similar observation can be made for transitions where the lower token
		goes from below a branch to above.
		
		We denote by $T_i$ the subtree of $T$ containing all configurations
		having a path to $\aconf{q_{i}}{u_{i}}{v_{i}}$.
		We aim to prove that $\runaprime$ reaches the position $i$,
		and the first time it does is in state $(\underline{q_i},\overline{q_i})$.
		
		We remark that for every cases \ref{item:1} to \ref{item:4}
		of the transition function,
		as long as the transition function is defined the run $\runaprime$
		can continue.
		As we only visit configurations of $T$, as long as $\rho$ is infinite,
		as assumed by Claim~\ref{claim:1dt2R}, so is $\runaprime$.
		
		Next, as $\Aa'$ is reversible, it cannot loop,
		as it would require two different configuration to go to
		the same one to enter the loop, which would break codeterminism.
		So $\runaprime$ starts in $T_i$, does not stop nor loops,
		so since $T_i$ is finite, $\runaprime$ has to end up leaving $T_i$. 
		So $\runaprime$ reaches position $i$, while only moving along $T_i$. 
		As $(\underline{q_i},\overline{q_i})$ is the only possible state
		where $\runaprime$ follows $T_i$ while having its first token
		above $\rho$ and the second below $\rho$,
		$\runaprime$ first reaches $i$ in state $(\underline{q_i},\overline{q_i})$.
		
		As this is true for any position $i$, and since $T_i$ contains $T_{i-1}$,
		we can conclude the proof of Claim~\ref{claim:1dt2R}.
	\end{proof}
	
	Based on this claim, we show that \(  \Lang{\Aa} = \Lang{\Aa'} \).
	Let \( w \) be an infinite word accepted by \( \Aa \),
	and \( \rho \) be the accepting run of \( \Aa \) on \( w \).
	We showed that the configurations \( \aconf{q}{u}{v} \)
	happen in the same order
	in \( \rho \) as configurations of the form
	\( \aconf{\left( \underline{q}, \overline{q} \right)}{u}{v} \) in \( \runaprime \),
	the run on \( w \) of \( \Aa' \).
	Moreover \( \cc'\left( \left( \underline{p}, \overline{p} \right), a \right) < \cc'\left( s, a \right) \)
	for all \( a \in \alp \) and for
	all \( s \) of another form,
	and as \( \left| \infi\left( \rho \right) \right| > 0 \)
	(because \( w \) is infinite and \( \Aa \) has a finite number of states),
	\( min\left\{ t | t \in \infi\left( \rho \right) \right\} =
	min\left\{ t | t \in \infi\left( \runaprime \right) \right\} \).
	So \( \Lang{\Aa'} = \Lang{\Aa} \).
	
	\noindent\textbf{Construction of \( \Tt' \).}
	Let \( \Tt' = \left( Q', \alp, \delta', q_0', \chi', \alpo, \lambda' \right) \)
	be the \( \rbt \) having \( \Aa' \) as underlying automaton,
	with
	\[
	\begin{aligned}
		\lambda'(s, a) = \begin{cases}
			\lambda(q, a) &\text{ if } s = \left( \underline{q}, \overline{q} \right) \\
			\varepsilon &\text{ if } s \text{ is of another form.}
		\end{cases}
	\end{aligned}
	\]
	Because we showed that states of the form
	\( \left( \underline{q}, \overline{q} \right) \) are
	met in the run of \( \Aa' \) on a given word in the same order as states \( q \)
	in the run of \( \Aa \) on the same word,
	the output of \( \Tt' \) is the same as the output of \( \Tt \).
	So we have that \( \sem{\Tt} = \sem{\Tt'} \).
	
	Finally, \( \left| Q' \right| = 4 \left| Q \right|^2 \), justifying the complexity.
\end{proof}

\section{From $\cbsst$ to $\rbt$}\label{sec:cpsst2rpt}

We extend another result from~\cite{DFJL17} about constructing a reversible two-way transducer from a Streaming String Transducer.
The prodcedure and the complexity are similar. 
The main difference is that the procedure only works
thanks to the distinguished register $\out$ of $\cbsst$.
Without it, production could depend on an infinite property of the input word,
which is not realizable by a deterministic (and hence reversible) machine.

\begin{theorem}\label{thm-SSTtoR2W}
  Let $\Tt$ be a $\cbsst$ with $n$ states and $m$ registers.
  Then we can construct a $\rbt$ $\Ss$ with $O(n^2m)$ states such that $\sem{\Tt}=\sem{\Ss}$.
\end{theorem}

\begin{proof}
	We prove that $\sem{\Tt}=\sem{\Ff}\circ\sem{\Dd}$ where $\Dd$ is a $\odbt$ and $\Ff$ is a $\rbt$.
	The transducer $\Dd$ has the same underlying automaton as $\Tt$, but instead of applying
	a substitution $\sigma$ to the registers, $\Dd$ enriches the input letter with $\sigma$.
	Then, the transducer $\Ff$ uses the flow of registers output by $\Dd$ to output the
	contents of the relevant registers in a reversible fashion.
	Finally, we construct a $\rbt$ $\Dd'$ equivalent to $\Dd$ by 
	Theorem~\ref{thm:odbtEqRdbt} and we obtain the desired $\rbt$ $\Ss=\Ff\circ\Dd'$ by 
	Theorem~\ref{thm:composition}.

	\noindent\textbf{Formal Construction.}
	Let $\Tt=(Q,\alp,\delta,q_0,\chi,\alpo,\R,\outreg,\lambda)$ be a $\cbsst$.

	We define the $\odbt$ by $\Dd=(Q,\alp,\delta,q_0,\chi,\subst{\alpo}{\R},\gamma)$ where 
	$\gamma(q,a)=\lambda(q,a)$.
	The reversible transducer is
	$\Ff=(Q',\subst{\alpo}{\R},\alpha,q_0',\emptyset,\alpo,\beta)$ where:
	\begin{itemize}
		\item $Q'=\R\times \{i,o\}$ with $Q^+=\R\times\{o\}$ and $Q^-=\R\times\{i\}$.  We will
		denote by $r_i$ (resp.\ $r_o$) the state $(r,i)$ (resp.\ $(r,o)$).
		Informally, being in state $r_i$ means that we need to compute the content of $r$,
		while state $r_o$ means we have just finished computing it.
		
		\item The initial state is $q_0'=\outreg_o$.
		
		\item There is no accepting condition;
		
		\item $\alpha$ and $\beta$ both read a state in $Q'$ and a substitution
		$\sigma\in\subst{\alpo}{\R}$, or the leftmarker $\leftend$, which is treated as a
		substitution $\sigma_{\leftend}$ associating $\epsilon$ to every register.  We
		define $\alpha$ and $\beta$ as follow:
		\begin{itemize}
			\item If the state is $r_i$ for some $r\in\R$.  
			\begin{itemize}
				\item If $\sigma(r)=v\in\alpo^*$ contains no register, then
				$\alpha(r_i,\sigma)=r_o$ and $\beta(r_i,\sigma)=v$.
				
				\item If $\sigma(r)=vs\gamma$ with $v\in\alpo^{*}$ and $s\in\R$ is the first
				register appearing in $\sigma(r)$, then $\alpha(r_i,\sigma)=s_i$ and
				$\beta(r_i,\sigma)=v$.
			\end{itemize}

			\item If the state is $r_o$ for some $r\in\R$.  
			Recall that from the definition of copyless SSTs, for any register $r$, there exists at 
			most one register $t$, such that $r$ occurs in $\sigma(t)$, and in this case $r$ 
			occurs exactly once in $\sigma(t)$.
			\begin{itemize}
				\item Suppose that for some register $s$ we have $\sigma(s)=\gamma rv$ with 
				$v\in\alpo^{*}$. 
				Then $\alpha(r_o,\sigma)=s_o$ and $\beta(r_o,\sigma)=v$.
				
				\item Suppose that for some register $t$ we have $\sigma(t)=\gamma rvs \gamma'$ with 
				$v\in\alpo^{*}$ and $s\in\R$. 
				Then $\alpha(r_o,\sigma)=s_i$ and $\beta(r_o,\sigma)=v$.

				\item If $r$ does not appear in any $\sigma(s)$, then the computation stops and
				rejects.  This somehow means that we are computing the contents of a register that
				is dropped in the original SST. This will not happen if what is fed to $\Ff$ is 
				produced by $\Dd$.
			\end{itemize}
		\end{itemize}
	\end{itemize}
	
	\noindent\textbf{Reversibility of $\Ff$.}
	The transducer $\Ff$ is clearly deterministic by construction.
	Let us prove that it is codeterministic.
	To this end, let $s_i$ be a state and $\sigma$ a substitution.  Looking at the
	transition function, its antecedent $\alpha^{-1}(s_{i},\sigma)$ is either $r_i$ if
	$\sigma(r)$ starts with $vs$ for some word $v\in \alpo^*$, or $r_o$ if there is some
	register $t$ that contains $rvs$ for some word $v\in \alpo^*$.
	Since we only consider copyless substitutions, there is at most one register that
	contains $s$.  The two options are then mutually exclusive, as one requires that $s$ be
	the first register to appear, and the second requires that there is a register before
	$s$.
	
	The proof for a state $s_o$ is similar.
	The antecedent $\alpha^{-1}(s_{o},\sigma)$ is either $s_i$ if $\sigma(s)$ contains no
	register, or $r_o$ if $r$ is the last register appearing in $\sigma(s)$.  Since these
	two are mutually exclusive, we get that $\Ff$ is reversible.

	\noindent\textbf{Correctness of the construction.}
	First, let us remark that the domain of $\Dd$ is the set of input words $u$ such that
	$\Tt$ has an infinite accepting run over $u$ since they share the same underlying automaton.
	So the domain of $\Tt$ is the set of words in the domain of $\Dd$ on which $\Tt$
	produces an infinite word.
	
	Let $(q_{0},\nu_{0})\xrightarrow{a_{0}}(q_{1},\nu_{1})\xrightarrow{a_{1}}(q_{2},\nu_{2})
	\xrightarrow{a_{2}}(q_{3},\nu_{3})\cdots$ be the accepting run of some input word 
	$u=a_{0}a_{1}a_{2}\cdots\in\alp^{\omega}$ in the domain of the $\cbsst$ $\Tt$.
	For $j\geq0$, let $\sigma_{j}=\lambda(q_{j},a_{j})$ so that 
	$\sem{\Dd}(u)=\sigma_{0}\sigma_{1}\sigma_{2}\cdots$.
	
	We prove by induction that for every position $j\geq0$ of $u$, the run of the transducer
	$\Ff$ on $\sem{\Dd}(u)$ reaches the state $\outreg_o$ in position $j$ having produced
	the content $\nu_{j}(\outreg)$
	of the run of $\Tt$ on $u$ up to position $j$.
	For $j=0$, there is nothing to prove as the registers are initially empty and
	$\outreg_o$ is the initial state of $\Ff$.
	
	Now suppose that the run of $\Ff$ on $\sem{\Dd}(u)$ reaches some position $j$ in state
	$\outreg_o$, having produced $\nu_{j}(\outreg)$.
	Recall that, by definition of $\lambda$, the substitution $\sigma_{j}=\lambda(q_{j},a_{j})$
	used by $\Tt$ at position $j$ is such that $\sigma_{j}(\outreg)=\outreg\cdot\gamma$.
	Then if there is no other register, i.e.,
	if $\gamma=v\in \alpo^*$, by definition of $\alpha$, $\Ff$ moves to $j+1$ in state
	$\outreg_o$ and produces $v$, so that its cumulated production is 
	$\nu_{j}(\outreg)\cdot v=\nu_{j+1}(\outreg)$.
	
	The interesting case is of course when some registers are flown to $\outreg$.
	Let $r$ be the second register of $\sigma_{j}(\outreg)$, i.e., $\sigma_{j}(\outreg)$
	starts with $\outreg \cdot v r$ with $v\in \alpo^*$.
	Then, by definition of $\alpha$ and $\beta$, $\Ff$ stays at position $j$ switching to
	state $r_i$ and producing $v$.
	Then, using Claim~\ref{rrRunofR}, $\Ff$ reaches $r_o$ at position $j$ producing the
	content of $\nu_{j}(r)$.
	We repeat this process to exhaust all registers appearing in $\sigma(\outreg)$, reaching
	finally state $\outreg_o$ at position $j+1$ with cumulated production 
	$\nu_{j}(\sigma_{j}(\outreg))=\nu_{j+1}(\outreg)$, proving the induction.
	
	\begin{claim}\label{rrRunofR}
		For all positions $j\geq0$ and registers $r\in\R$, there exists a right-to-right run
		$(r_i,r_o)$ of $\Ff$ starting and ending at position $j$ and which produces the
		content of $\nu_{j}(r)$.
	\end{claim}
	
	\begin{proof}[Proof of Claim~\ref{rrRunofR}.]
		The proof is by induction on $j$.
		If $j=0$ then $\nu_{0}(r)=\varepsilon$ and the run of $\Ff$ starting at position $0$
		in state $r_{i}$ reads $\sigma_{\leftend}$.  By definition of $\alpha$ and $\beta$ the
		run produces $\sigma_{\leftend}(r)=\varepsilon$ and switches from $r_i$ to $r_o$,
		proving the claim for $j=0$.
		
		Now assume that the claim is true for $j$.  Consider the run $\rho$ of $\Ff$ starting
		in state $r_{i}$ at position $j+1$.  The run $\rho$ starts by reading $\sigma_{j}$.
		If $\sigma_{j}(r)=v\in\alpo^{*}$ then the run produces $v=\nu_{j+1}(r)$ and switches
		from $r_i$ to $r_o$, proving the claim.  The second case is when $\sigma_{j}(r)$
		starts with some $vs$ with $v\in\alpo^{*}$ and $s\in\R$.  Then, the first transition
		of $\rho$ produces $v$ and moves to position $j$ in state $s_{j}$.
		By induction hypothesis, there is a right-right $(s_i,s_o)$-run starting at $j$ and
		producing $\nu_{j}(s)$.  Then, the run reads $\sigma_{j}$ in state $s_o$ and, either
		goes $r_o$ in position $j+1$ producing $v'\in\alpo^{*}$ if $\sigma_{j}(r)$ ends with
		$sv'$ ($s$ is the last register flown to $r$), or goes to $t_i$ in position $j$
		producing $v'\in\alpo^{*}$ if $\sigma_{j}(r)$ contains the factor $sv't$ ($t$ is
		the next register flown to $r$).
		By iterating this process again, we exhaust the registers flown to $r$, produces their
		content meanwhile.  Finally, the run $\rho$ ends in position $j+1$ with state $r_o$
		and has produced $\nu_{j+1}(r)=\nu_{j}(\sigma_{j})(r)$, proving the claim.
	\end{proof}
	
	Coming back to the proof of correctness, we have shown that for all positions $j\geq0$, 
	$\Ff$ has an initial run on $\sem{\Dd}(u)$ reaching position $j$ in state $\outreg_{o}$ 
	and producing $\nu_{j}(\outreg)$. This proves that $\sem{\Dd}(u)$ is in the 
	domain of $\Ff$ (the maximal initial run of $\Ff$ on $\sem{\Dd}(u)$ reads the whole 
	word and $\Ff$ accepts only if $\Tt$ produces infinitely often) and $\sem{\Ff}(\sem{\Dd}(u))=\bigsqcup_{j\geq0}\nu_{j}(\outreg)=\sem{\Tt}(u)$.
	Therefore, $\sem{\Tt}=\sem{\Ff}\circ\sem{\Dd}$.
	
	Finally, we construct a $\rbt$ $\Dd'$ equivalent to $\Dd$ by 
	Theorem~\ref{thm:odbtEqRdbt} and we obtain the desired $\rbt$ $\Ss=\Ff\circ\Dd'$ by 
	Theorem~\ref{thm:composition}.
\end{proof}

\section{From $\dbt$ to $\cbsst$}\label{sec:dpt2cpsst}
The construction presented in this section is the most involved of the paper. It is adapted from~\cite{DJR16}. 
Given a deterministic two-way transducer, we construct a $\cbsst$ that realizes the same function. Here again, the main complications from infinite words are dealing with the acceptance condition and the impossibility to get the final configuration of the run.

\begin{theorem}\label{thm:2DTtoSST}
  Given a $\dbt$ $\Tt$ with $n>0$ states and $k$ coloring functions over $\ell$ colors, we
  can construct a $\cbsst$ $\Ss$ with $O(\ell^{kn}(2n)^{2n})$ states, $2n-1$ registers and
  $k$ coloring functions over $\ell$ colors such that $\sem{\Tt}=\sem{\Ss}$.
\end{theorem}

\begin{proof}[Sketch of proof.]
  We improve on the classical Shepherdson construction~\cite{Shepherdson59} from two-way machines to one-way.
  In this construction, the one-way machine computes information about the runs of the
  two-way machine on the prefix read up to the current position.  More precisely, it
  stores the state reached on reading the prefix starting at the initial state, as well as
  a succinct representation of the information about all right-right runs.  Further, by
  associating a register to each run in this representation, we can construct an SST
  equivalent to the two-way machine.
  
  Upon reading some letter $a\in\alp$, the prefix $w$ we are interested in grows to $wa$,
  and consequently, we have to update the information about the right-right runs.  While
  some right-right runs on the prefix $w$ may be extended to right-right runs on $wa$,
  some runs (which cannot be extended) may die, and hence needs to be removed.  A third
  possiblity is that, upon reading a letter $a$, some right-right runs may merge. This implies that the construction would not be copyless, as if two right-right runs over $ua$ are the extension of a same right-right run over $u$, the register storing the production of the run over $u$ needs to be copied in both runs over $ua$.

  In order to compute a copyless SST, we improve this construction by refining the
  information stored by the one-way machine: it stores not only the set of right-right
  runs, but also whether they merge and the respective order of the merges.  Essentially,
  the latter representation keeps track of the \emph{structure} of the right-right runs on
  the prefix read up to the current position, as well as the \emph{output} generated by
  these runs.  The resulting information can be represented as a forest, which is a
 (possibly empty) set of trees.  Then we associate a register to each edge of
  the forest, so that the update function can be made copyless.  The number of registers
  required is still linear in the number of states.
\end{proof}
\begin{wrapfigure}[17]{r}{.35\textwidth}
	\centering
  \includegraphics[width=.35\textwidth]{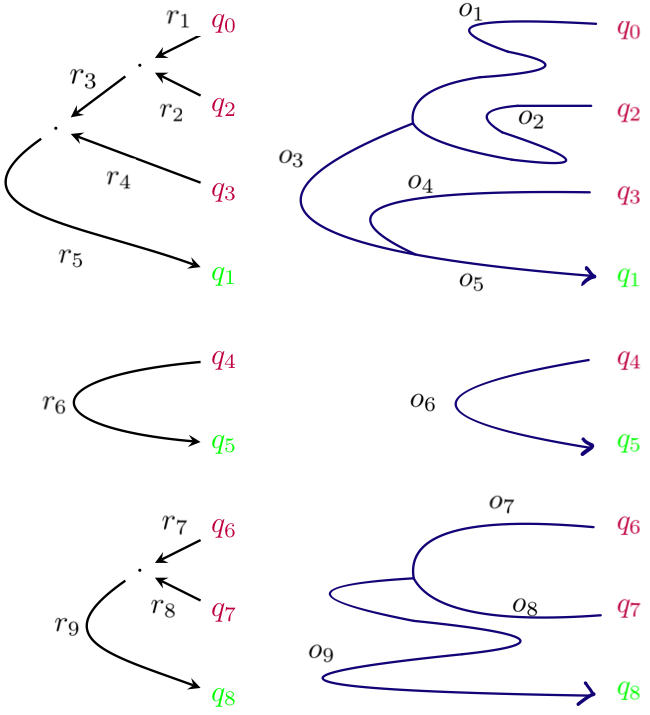}

  \caption{Example of a merging forest (on the left) corresponding to right-right runs (on the right) of a two-way machine for some prefix $u$ of an input word.  
  }
	\label{fig:merging-forest}
\end{wrapfigure}

Formally, we call the structure used to model the right-right runs \emph{merging forests}, which we define as follow:
\begin{definition}\label{def:mforests}
  Given a set of states $Q=Q^{+}\uplus Q^{-}$, a set of coloring functions $\chi$
  and an integer $\ell>0$,
  we define the \emph{merging forests} on $(Q,\chi,\ell)$, denoted $\mf{Q}$, as the set of forests $F$
  such that:
  \begin{itemize}[nosep]
    \item the leaves of all trees of $F$ are labeled by distinct elements of $Q^-$,

    \item the roots of all trees of $F$ are labeled by distinct elements of $Q^+$,

    \item all unary nodes (exactly one child) are roots.
  \end{itemize}
  The leaves are also labeled by a $\chi$-tuple of integers less than $\ell$.
\end{definition}

Informally, an element $F\in\mf{Q}$ describes a set of right-right runs, such that if $q$
is the root of a tree and $p$ is one of its leaves, then $(p,q)$ is a right-right run.
Notice that, if two leaves $x$ and $y$ belong to the same tree, then the right-right runs
starting in $x$ and $y$ will merge.
The structure of the tree reflects the order in which the runs sharing the same root merge.
The integer labels of leaves serve as coloring for the parity acceptance condition.
An example of a merging forest is depicted in Figure~\ref{fig:merging-forest}.

Note that in Figure~\ref{fig:merging-forest}, the states in $Q^{-}$ are depicted in purple, while the states in $Q^{+}$ are
depicted in green.
The forest comprises of 3 trees - one with root $q_1$ and leaves $q_0, q_2$ and $q_3$, the second with root $q_5$ and leaf $q_4$, and the third with root $q_8$ and leaves $q_6$ and $q_7$.
For each tree, there are right-right runs from its leaves to the root.
For instance, from the merging forest in Figure~\ref{fig:merging-forest}, we can infer
that there are three right-right runs entering $u$ on the right in states $q_0$, $q_2$
and $q_3$ respectively, and emerging out of $u$ in state $q_1$, moreover the run
starting from $q_{0}$ first merges with the run starting from $q_{2}$ and these two runs
then merge with the one starting from $q_{3}$.
The figure also depicts the register corresponding to the edges of the forest. 
Further, the output generated by the part of the run labeled $o_1$ is stored in
the register $r_1$, the output for the part labeled $o_2$ is stored in $r_2$, $o_3$ in $r_3$ and $o_5$ in $r_5$.

We are now ready to give the formal construction of Theorem~\ref{thm:2DTtoSST}.
We defer the proof of correctness to later in this section.
\begin{proof}[Proof of Theorem~\ref{thm:2DTtoSST}.]
  Given a $\dbt$ $\Tt=(Q, \alp,\delta,q_0,\chi,\alpo,\lambda)$, we construct 
  a $\cbsst$ $\Ss=(Q',\alp,\alpha,q_0',\chi',\alpo,\R,\outreg,\beta)$ such that:
  \begin{itemize}[nosep]
    \item $Q'=Q\times\mf{Q}$,
    
    \item $q_0'=(q_{0},F_0)$ where $F_0$ is the forest having only leaves and roots and
    edges from a leaf $u$ labeled $p\in Q^-$ to a root $v$ labeled $q\in Q^+$ if
    $\delta(p,\leftend)=q$. 
    
    \item $\R$ is a set of registers of size $2|Q|-1$ with a distinguished register
    $\outreg$.
    
    \item the coloring functions $\chi'=\{\cc'\mid \cc\in\chi\}$ are described below.
    
    \item the definitions of $\alpha$ and $\beta$ are more involved and given below.
  \end{itemize}
  
For each merging forest $F\in\mf{Q}$, we fix a map $\xi_{F}$ associating distinct
registers from $\R\setminus\{\outreg\}$ to edges of $F$.  This is possible since the
number of edges in $F$ is at most $2|Q|-2$ (see Lemma~\ref{lem-sizeOfMF}).

For simplicity sake, we assume that
for each edge $(u,v)$ of $F_{0}$ corresponding to transition $\delta(p,\leftend)=q$, the
register $\xi_{F_{0}}(u,v)$ is initialized with the production $\lambda(p,\leftend)$. Note that considering initialized registers does not add expressiveness, as it could be simulated using a new unreachable initial state.
Other registers are initially empty.

The definitions of the transition function $\alpha$ and the update function $\beta$ are
intertwined.
Let $(q,F)\in Q\times\mf{Q}$ be a state of $\Ss$ and $a$ a letter of $\alp$.
We describe the state $(p,F')=\alpha((q,F),a)$.
First we construct an intermediate \emph{graph} $G$ that does not satisfy the criteria of
the merging forests, then explain how $G$ is transformed into a merging forest
$F'\in\mf{Q}$. The steps of the construction are depicted in Figure~\ref{fig:merging-forest-final}.

\smallskip\noindent\textbf{Construction of G.} %
The graph $G$ is built from $F$ as follows.  
First, we add new isolated nodes $Q_{c}=\{q_{c}\mid q\in Q\}$ to the forest $F$.  
These nodes will serve as the roots and leaves of $F'$.
For each transition $\delta(p,a)=q$, we will add an edge connecting these new
nodes and the roots and leaves of $F$.  We also extend the labelling $\xi_{F}$ to a
labelling $\xi_{G}$ by adding productions $\lambda(p,a)$ to the new edges of $G$.

Let $p\in Q$ and let $q=\delta(p,a)$.  
If $p\in Q^+$ is the label of the root $u$ of some tree in $F$, we add an edge from $u$ to either $q_c$ if $q\in Q^+$, or to $v$ if $q\in Q^-$ and there exists a leaf $v$ labeled $q$ in $F$.
If $p\in Q^{-}$, we add an edge from $p_c$ to either $q_c$
if $q\in Q^+$, or to $v$ if $q\in Q^-$ and there exists a leaf $v$ labeled $q$ in $F$.
The added edge is labeled $\lambda(p,a)$ by $\xi_{G}$.
The construction is illustrated in the first two steps of Figure~\ref{fig:merging-forest-final}.
Note that $G$ may now have cycles.

\smallskip\noindent\textbf{The new state $p$.} %
To compute the first component of the new state of $\Ss$, let $r=\delta(q,a)$.  If $r$ belongs to $Q^+$, then $p=r$ and  $\lambda(q,a)$ is appended to the output register:
$\beta((q,F),a)(\outreg)=\outreg\cdot\lambda(q,a)$. 
Moreover, the colors are directly inherited: $\bar{\cc}((q,F),a)=\cc(q,a)$ for 
all $\cc\in\chi$.
Otherwise, assume that there is a leaf $u$ in $F$ labeled $r\in Q^{-}$ (if not, the
transition is undefined). We consider the maximal path in $G$ starting from $u$. If 
this path is looping or if it ends in a root of $F$ then the transition is undefined.
Otherwise, it ends in a new node $v$ of $G$.
Let $p\in Q^{+}$ be the state such that $v=p_{c}$.
We append to the $\outreg$ register first $\lambda(q,a)$ and then the $\xi_{G}$
labels of the edges of the path from $u$ to $v$ in $G$, in the order of the path.
These $\xi_{G}$ labels are either registers given by $\xi_{F}$ for edges of $F$, or local
outputs of the form $\lambda(x,a)$ for the new edges, i.e., those in $G \setminus F$.
Moreover, for each $\cc\in\chi$, we let $\cc'((q,F),a)$ be the minimum of (1) the
$\cc$-values of the labels of leaves of $F$ appearing in the path from $u$ to $v$ in
$G$, and (2) the $\cc$-values of the transitions used to create the new edges of $G$ in
this path.

The third figure of Figure~\ref{fig:merging-forest-final} illustrates the case  where $\delta(q,a)=q_4\in Q^{-}$. We then look at the path starting in $q_4$, which leads to the state $(q_8)_c$ after reading $a$.

\smallskip\noindent\textbf{Construction of $F'$ and the update of registers other than $\outreg$.} %
This is illustrated by the third and fourth figures of Figure~\ref{fig:merging-forest-final}.
We erase all nodes (and adjacent edges) that are on a cycle of $G$ since such 
cycles cannot be part of an accepting run of $\Tt$ (see the part of $G$ boxed in blue). The resulting graph is now a acyclic. 
We also erase all nodes and edges which are not on a path from a leaf in $Q_{c}^{-}$ to a 
root in $Q_{c}^{+}$ since they cannot be part of a right-right run on $ua$ (see the parts boxed in yellow).
Finally, we erase the tree with root $v=p_c$ in the second case of the definition of $p$
above.  Indeed, should any right-right run simulated by this tree appear later, the
resulting run would loop on a finite prefix of the input (see the part boxed in green).

\begin{figure}[tbp]
	\centering
	\begin{subfigure}{\textwidth}
    \centering
		\includegraphics[width=14cm]{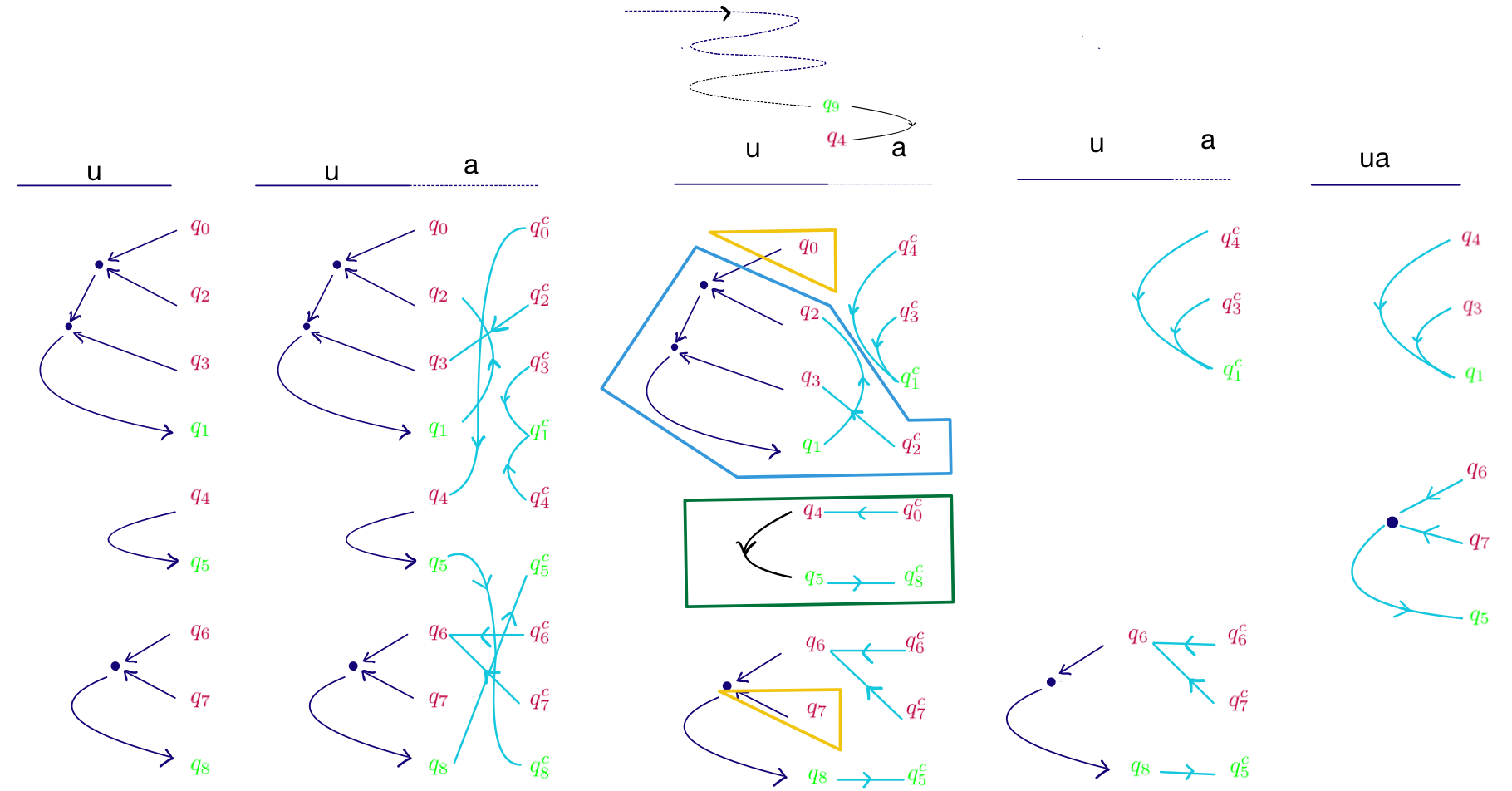}
	\end{subfigure}
	\caption{
Illustration of the procedure to calculate $F'$ from $F$ and a letter $a$.\\
The first figure is $F$, the second is the graph $G$ built from $F$ and the transition occuring at $a$. 
The third figure illustrates the trimming done from $G$ to obtain a forest depicted in the fourth figure.
Finally, we merge unary internal nodes to obtain the merging forest of the last figure.
	}
	\label{fig:merging-forest-final}
\end{figure}

The resulting forest has leaves in $Q_{c}^{-}$ and roots in $Q_{c}^{+}$.  Each remaining
new leaf or root $q_{c}\in Q_{c}$ is labeled $q$, i.e., by dropping the $c$ index.
A new leaf $q_{c}\in Q^{-}$ is labeled by a $\chi$-tuple $(m_{\cc})_{\cc\in\chi}$
of integers less than $\ell$ where $m_{\cc}$ is the minimum of (1) the $\cc$-component of
the labels of leaves of $F$ appearing in the branch in $G$ from $q_c$ to its root, and (2)
the values $\cc(p,a)$ of the transitions used to create the new edges of $G$ in this
branch.
We forget the initial labeling of leaves and roots of $F$.

To obtain a merging forest, it remains to remove unary internal nodes. This is shown between the fourth and fifth figures of Figure~\ref{fig:merging-forest-final}.
We replace each \emph{maximal} path $\pi=u_{0},u_{1},u_{2},\cdots,u_{n-1},u_{n}$
($n\geq1$) with $u_{1},\ldots,u_{n-1}$ unary nodes by a single edge $e=(u_{0},u_{n})$.
We obtain the merging forest $F'$.
The update function is simultaneously defined by
$$
\beta((q,F),a)(\xi_{F'}(e))=\xi_{G}(u_{0},u_{1})\xi_{G}(u_{1},u_{2})\cdots
\xi_{G}(u_{n-1},u_{n}) \,.
$$
Notice that for each remaining edge $f$ of $G$ we have either $f\in F$ and 
$\beta((q,F),a)(\xi_{F'}(f))=\xi_{G}(f)=\xi_{F}(f)$
or $f$ is a new edge and $\beta((q,F),a)(\xi_{F'}(f))$ is set to some $\lambda(s,a)$.  Since each edge $f$ of $F$
contributes to at most one edge $e$ of $F'$, or to the path from $u$ to $v=p_{c}$ which
flows into the $\outreg$ register (but not both), it implies that the
substitution $\beta((q,F),a)$ is copyless.\qedhere

\end{proof}

To prove correctness of Theorem~\ref{thm-SSTtoR2W},
we need the following lemma, relying on Cayley's Formula to count the number of merging forests.
\begin{lemma}\label{lem-sizeOfMF}
	Let $Q=Q^{+}\uplus Q^{-}$ be of size $n>0$ with $Q^{+}\neq\emptyset$, $\chi$
	of size $k\geq0$ and $\ell>0$.  Then each element $F$ of $\mf{Q}$ has at most $2n-2$
	nodes, $2n-2$ edges; and $\mf{Q}$ itself is of size at most
	$\ell^{k(n-1)}(2n-1)^{2n-3}$.
\end{lemma}
\begin{proof}
	We first compute the maximal number of nodes and edges in a nonempty forest $F$ of $\mf{Q}$.
	Note that, since $F$ is nonempty, both $Q^{+}$ and $Q^{-}$ must be nonempty.
	Let $\mathsf{nbe}(t)$ and $\mathsf{nbl}(t)$ be the number of edges and leaves of the 
	tree $t$. If $t$ has no unary nodes, then $\mathsf{nbe}(t)\leq 2\mathsf{nbl}(t)-2$.
	Since in a forest $F\in\mf{Q}$, all unary nodes are roots, it follows that 
	$\mathsf{nbe}(t)\leq 2\mathsf{nbl}(t)-1$ for all trees $t\in F$.
	Also, the number of nodes of a tree is $\mathsf{nbn}(t)=1+\mathsf{nbe}(t)$.
	We deduce that 
	\begin{align*}
		\mathsf{nbe}(F) &= \sum_{t\in F}\mathsf{nbe}(t)\leq\sum_{t\in F} 2\mathsf{nbl}(t)-1
		\leq 2|Q^{-}|-1 \leq 2n-3 
		\\
		\mathsf{nbn}(F) &= \sum_{t\in F}\mathsf{nbn}(t)\leq\sum_{t\in F} 2\mathsf{nbl}(t)
		\leq 2|Q^{-}| \leq 2n-2 \,.
	\end{align*}

	Now we compute the size of the set $\mf{Q}$.
	Cayley's formula~\cite{Cayley1889} states that the number of non
	oriented trees with $m$
	differently labeled nodes
	is $m^{m-2}$.
	The difference here is that first we deal with forests with at most $2n-2$ nodes, and
	secondly only the leaves and roots are labeled.
	The first point can be dealt with by adding a new node as the root of all trees of the
	forest,and new nodes if needed to get exactly $m=2n-1$ nodes in the tree.
	For the second point, we can label arbitrarily the remaining nodes.
	Finally, as each leaf is labeled by a $\chi$-tuple of integers less than $\ell$,
	each tree can appear in $\mf{Q}$ up to $(\ell^k)^{|Q^{-}|}$ many times.
	The size of $\mf{Q}$ is then smaller than $\ell^{k(n-1)}(2n-1)^{2n-3}$.
\end{proof}

We can now prove correctness for Theorem~\ref{thm:2DTtoSST}.

\begin{proof}[Proof of Correctness for Theorem~\ref{thm:2DTtoSST}.]
	We begin with the size of \( \Ss \).
	Using Lemma~\ref{lem-sizeOfMF}, we get that $|Q'|\leq 
	n(\ell^{k(n-1)}(2n-1)^{2n-3}=\mathcal{O}(\ell^{kn}(2n)^{2n})$.
	Using carefully the registers, we only ever need at most $2n-1$ registers.

	\smallskip\noindent\textbf{Proof of correctness.}
	First given a finite word $w$, we say that a right-right run $(x,y)\in Q^{-}\times Q^{+}$
	of the $\dbt$ $\Tt$ on $w$ is useful if there exists an infinite word $w'$ such that the
	run of $\Tt$ on $ww'$ is accepting and reaches $x$ on position $|w|$.

	We first prove that the state of the $\cbsst$ contains all the needed information, then
	prove that the registers can be used to produce the output.  We prove by induction on the
	size of a word $w$ that the state $(q,F)$ of the constructed $\cbsst$ reached after
	reading $w$ is such that $(q_{0},q)$ is a left-right run on $w$ and $F$ contains
	information about all useful runs on $w$.  Moreover, the $\outreg$ register contains the
	production of the left-right run $(q_{0},q)$ on $w$ and, given a path $\pi=u_0,\ldots,u_n$
	in $F$ from a leaf $u_0$ labeled by $x$ to a root $u_n$ labeled by $y$, the production of
	the right-right run $(x,y)$ is given by the concatenation of the registers
	$\xi_F((u_0,u_1))\ldots\xi_F((u_{n-1},u_n))$.

	First, if $w$ is empty, then the initial state is $(q_{0},F_{0})$ where $F_{0}$ describes
	the set of all right-right runs on $\leftend$.  The register $\outreg$ is empty and each
	tree in the forest $F_{0}$ is reduced to a single edge containing the associated
	production, hence proving the initial case.

	Now suppose that the statement holds for some word $w$ and some state $(q,F)$ and let
	$a\in\alp$ be a letter.  We prove the statement for $wa$.  Let $(p,F')=\alpha((q,F),a)$.
	If $p=\delta(q,a)\in Q^{+}$, then $(q_{0},p)$ is a left-right run on $wa$ and
	$\beta((q,F),a)(\outreg)=\outreg\cdot\lambda(q,a)$, corresponding to the claim for the
	left-right run.  Otherwise, let $r=\delta(q,a)\in Q^-$.  The state $p$ is then described
	in $G$ as the state reached by the maximal path in $G$ from the leaf $u$ of $F$ labeled by
	$r$, to the new node $p_c$.  Using the induction hypothesis, $F$ describes the useful
	right-right runs on $w$.  Then, following the maximal path from $u$ in $G$, we see the
	sequence of right-right runs on $w$ and left-left runs on $a$, up to the last left-right 
	transition on $a$ leading to state $p$.  
	This means that we have computed $p$ such that $(q_{0},p)$ is the left-right run on $wa$.
	We also append to the register $\outreg$ all $\xi_G(e)$ for edges $e$ in the path.  By
	induction hypothesis, the registers contain the production of the useful right-right runs
	on $w$, and the added edges contains the local production, proving the claim for the
	left-right run.

	We now prove that all useful runs of $wa$ are in $F'$.  Let $(x,y)\in Q^{-}\times Q^{+}$
	be such a run.  Then either $\delta(x,a)=y$ and this edge is added in $G$ and remains in
	$F'$, or $(x,y)$ is a sequence starting with a rigth-left transition over $a$, then useful
	right-right runs over $w$ and left-left transitions over $a$, and finally a left-right
	transition over $a$.  In the first case, the register $\xi_{F'}((x,y))$ takes the label of
	$G$, i.e. the local production $\lambda(x,a)$, satisfying the claim as the path is reduced
	to a single edge.  In the second case, let $u_0,\ldots,u_n$ be the path from $x$ to $y$ in
	$G$.  Each edge $(u_i,u_{i+1})$ is either a new edge whose label is a local production, or
	a single edge of $F$.  The associated sequence of registers contains then all the output
	information of the $(x,y)$ run.  
	When reducing $G$ to $F'$, as only non branching paths of $G$ can be reduced to a single
	edge, there is no loss of information.  
	Finally, notice that the edges deleted from $G$ to obtain $F'$ are the ones that are not
	part of a usuful right-right run on $wa$, or are merging with the left-right run.  These
	latter right-right runs are not useful for $wa$: they cannot occur anymore in an accepting
	run of $\Tt$ since they would induce a loop on a finite prefix of the input.  Hence all
	edges required for the path from $x$ to $y$ appear in $F'$.  Consequently, there is no
	loss of run nor information, the path from $x$ to $y$ in $F'$ exists and the associated
	sequence of registers contains the production of the run.

	Finally, to prove that both transducers have the same domain, we remark that given the
	previous induction, if $(q,F)$ is the state reached by $\Ss$ after reading an input $w$,
	then upon reading a letter $a$, the color of the transition $\alpha((q,F),a)=(p,F')$ is
	the minimum of the color of all transitions used when extending the left-right run 
	$(q_{0},q)$ of $\Tt$ on $w$ to the left-right run $(q_{0},p)$ on $wa$.
	Then given an infinite word $u$, $\Ss$ has an infinite run on $u$ if and only if $\Tt$
	does, and the minimum of all colors appearing infinitely often is the same on both runs.
\end{proof}

\section{Continuity and topological closure of a $\dbt$}
\label{sec:cont}
The classical topology on infinite words (see e.g.~\cite{Perrin-Pin-Infinite-words}) defines the distance between two infinite words $u$ and $v$ as $d(u,v)=2^{-|u\wedge v|}$, where $u\wedge v$ is the longest common prefix of $u$ and $v$.
Then a function $f\colon \alp^\omega\to \alpo^\omega$ is continuous at $x\in \dom{f}$ if 

$$
\forall i\geq 0, \exists j\geq 0 \forall y\in \dom{f},|x\wedge y|\geq j \implies |f(x)\wedge f(y)|\geq i$$

A function $f$ is continuous if it is continuous at every $x\in \dom{f}$. We refer to~\cite{DFKL22} for more details.

Since a $\dbt$ is in particular deterministic, it realizes a continuous function. 
Indeed, the longer two input words share a common prefix, the longer their output will also do.

By comparison, a non-deterministic transducer can make choices depending on an infinite property of the input, e.g. whether there is an infinite number of $a$s, and thus realize a noncontinuous function. 

The question of characterizing the continuous functions realizable by transducers was studied in~\cite{CD22,CDFW23}.
In~\cite{CD22}, it was proved that for any continuous function realized by a non-deterministic one-way transducer $T$, there exists a deterministic two-way transducer $S$ such that for any $u\in \dom{T}$, $\sem{T}=\sem{S}$. 
This means that if a non-deterministic one-way transducer realizes a continuous function, although it can uses the non-determinism to refine its domain, the continuity property forbids it from producing non-deterministically.

In~\cite{CDFW23}, it was conjectured that the class of deterministic regular functions, i.e. functions defined by deterministic two-way transducers with a \buchi acceptance condition,
corresponds to the class of functions realized by a non-deterministic two-way transducer, called regular functions, that are continuous.

Since the class of $\dbt$ is strictly more expressive than the class of deterministic regular functions of~\cite{CDFW23}, but still realize continuous functions, the conjecture of~\cite{CDFW23} fails.
One can for example consider the function $f$ such that $f(u)=u$ if $u$ has a finite number of $a$s and is undefined otherwise. Then $f$ is continuous, as it is continuous on its domain, but it cannot be realized by a deterministic two-way transducer with a \buchi condition.

However, the refinement here only acts on the domain of the function, and not the production.
Indeed, let $T$ be a $\dbt$ which realizes a function $f$. 
We can define the topological closure of $\dom{T}$, which we denote $\widehat{\dom{T}}$ as the words $u$ such that there exists a sequence $(u_i)_{i\geq 1}$ where for all $i$, $u_i\in\dom{T}$ and
$$
\forall n,\exists i\forall j\geq i,\ |u\wedge u_j|\geq n
$$
We say that the sequence $(u_i)_{i\geq 1}$ converges to $u$.
Note that if two sequences $(u_i)_{i\geq 1}$ and $(v_j)_{j\geq 1}$ belong to $\dom{T}$ and converge to the same word $w$, then the elements will share longer and longer prefixes.
Since $T$ is deterministic, both sequences will then also produce words that share longer and longer prefixes.
Thus we can define $\hat{f}$, whose domain is $\widehat{\dom{T}}$, where $\hat{f}(u)$ is the limit of the images of any sequence $(u_i)_{i\geq 1}$ that belong to $\dom{T}$ and which converges to $u$.
The function $\hat{f}$ is in fact realized by the transducer $T$ where the accepting condition is dropped.
The domain of its underlying automaton is then a closed set in the classical topology  over infinite words, as it is recognized by a deterministic \buchi automaton where all transitions are final (see~\cite[Proposition 3.7, p. 147]{Perrin-Pin-Infinite-words}).
Note however that since the semantics of our transducers requires an input to produce an infinite word to be in the domain of the transducer, the domain of $\hat{f}$ might not be a closed set.

\paragraph*{Reversible two-way transducers with no accepting condition.}
Following this, let us consider the class of reversible transducers with no accepting condition ($\rt$), which is equivalent to saying all transitions are final within a \buchi condition.

The constructions presented in this article get simpler, with a better complexity:

\begin{theorem}
Given a deterministic two-way transducer $T$ with $n$ states and no accepting condition, we can construct a $\rt$ $S$ with $O(n^{4n+1})$ states such that $\sem{T}=\sem{S}$. 
\end{theorem}
\begin{proof}
Let $T$ be a deterministic two-way transducer with $n$ states.
Since there is no accepting condition, the construction from the proof of Theorem~\ref{thm:2DTtoSST} drops the arrays of integers from the merging forest.
Then the number of merging forests is in $O(n^{2n})$, and applying Theorem~\ref{thm:2DTtoSST} results in an SST $T'$ with $O(n^{2n})$ states and $2n$ variables.
Hence, by applying Theorem~\ref{thm-SSTtoR2W} to $T'$, we get an equivalent $\rt$ $S$ with $O(n^{4n+1})$ states.\qedhere
\end{proof}

Surprisingly, the class of $\rt$ has the expressive power of deterministic \buchi two-way transducers.
This is due to the fact that for an input to be accepted, it requires the corresponding output to be infinite. We can \emph{hide} the \buchi acceptance condition in this restriction.

The following theorem proves that starting from a reversible \buchi transducer,
we can construct one which does not use any acceptance condition.
To notice that reversible and deterministic \buchi transducers have the same expressive power,
one can rely on our main theorem, specializing it to a \buchi condition.
Indeed, a \buchi acceptance condition corresponds to a unique coloring function associating $0$
to accepting transitions and $1$ otherwise.
Then if we are given a deterministic \buchi two-way transducer,
we are able to produce a reversible one using a coloring function on the same domain,
hence a reversible \buchi two-way transducer.

\begin{theorem}\label{thm:NoAccIsBuchi}
	Let $T$ be a reversible \buchi two-way transducer with $n$ states.
	We can construct an equivalent $\rt$ $S$ with $3n$ states such that
	$\sem{T}=\sem{S}$.
\end{theorem}
\begin{proof}
	Let $T$ be a reversible \buchi two-way transducer with $n$ states.
	We define the $\rt$ $S$ similarly to $T$, but with three \emph{modes} of operation
	(and hence thrice the states): \emph{simulation}, \emph{rewind} and \emph{production}.
	The transducer $S$ starts by simulating $T$ without producing anything.
	Upon reaching an accepting transition $t$,
	it switches to rewind mode to and unfolds the run of $T$
	back to the previous accepting transition (or the start of the run).
	It finally follows the run of $T$ and produces the corresponding output up to seeing
	the accepting transition $t$, at which point it restarts the simulation.
	
	Then for a given word $u$, if $u$ belongs to the domain of $T$,
	the run of $T$ over $u$ will see accepting transitions infinitely often,
	and hence $S$ will go to production mode infinitely often too.
	Conversely, if $u$ does not belong to the domain of $T$,
	it means that either the run of $T$ over $u$ only sees a finite number of accepting transitions,
	or does not produce an infinite word. 
	In the latter case, $S$ will also produce a non finite word.
	In the former case, $S$ will remain in simulation mode and never produce anything,
	and thus $u$ will not be in the domain of $S$.
	
	We can remark that $S$ is reversible if $T$ is since:
	\begin{itemize}[nosep]
		\item within modes, transitions of $S$ are transitions of $T$,
		\item The transitions from one mode to another happen on every accepting transition.
		Hence transitions that come to a given mode from another are exactly the ones that exit it,
		so two transitions cannot go to a same state upon reading the same letter.
	\end{itemize}
\end{proof}

\section{Conclusion}~\label{sec:conclusion}
The main contribution of this paper is the result which shows that
deterministic two-way transducers over infinite words with a
generalized parity acceptance condition are reversible.
We also show that reversible two-way transducers over infinite words are closed under composition.
Our results can help in an efficient construction of two-way reversible transducers
from specifications presented as RTE~\cite{DGK-lics18} or
SDRTE ~\cite{DBLP:conf/lics/DartoisGK21} over infinite words. 
Earlier work~\cite{DBLP:conf/lics/DartoisGGK22}
in this direction on finite words relied on an efficient translation from
non-deterministic transducers used in parsing specifications to reversible ones;
our results can hopefully help extend these to infinite words. 

\bibliographystyle{abbrv}
\bibliography{Rev}

\begin{thebibliography}{10}

\bibitem{AFT12}
R.~Alur, E.~Filiot, and A.~Trivedi.
\newblock Regular transformations of infinite strings.
\newblock In {\em Proceedings of the 27th Annual {IEEE} Symposium on Logic in
  Computer Science, {LICS} 2012, Dubrovnik, Croatia, June 25-28, 2012}, pages
  65--74. {IEEE} Computer Society, 2012.

\bibitem{AlurFreilichRaghothaman14}
R.~Alur, A.~Freilich, and M.~Raghothaman.
\newblock Regular combinators for string transformations.
\newblock In T.~A. Henzinger and D.~Miller, editors, {\em Joint Meeting of the
  23rd {EACSL} Annual Conference on Computer Science Logic {(CSL)} and the 29th
  Annual {ACM/IEEE} Symposium on Logic in Computer Science (LICS), {CSL-LICS}
  '14, Vienna, Austria, July 14 - 18, 2014}, pages 9:1--9:10. {ACM}, 2014.

\bibitem{BR-DLT18}
N.~Baudru and P.-A. Reynier.
\newblock From two-way transducers to regular function expressions.
\newblock In M.~Hoshi and S.~Seki, editors, {\em 22nd International Conference
  on Developments in Language Theory, {DLT} 2018}, volume 11088 of {\em Lecture
  Notes in Computer Science}, pages 96--108. Springer, 2018.

\bibitem{CD22}
O.~Carton and G.~Dou{\'{e}}neau{-}Tabot.
\newblock Continuous rational functions are deterministic regular.
\newblock In S.~Szeider, R.~Ganian, and A.~Silva, editors, {\em 47th
  International Symposium on Mathematical Foundations of Computer Science,
  {MFCS} 2022, August 22-26, 2022, Vienna, Austria}, volume 241 of {\em
  LIPIcs}, pages 28:1--28:13. Schloss Dagstuhl - Leibniz-Zentrum f{\"{u}}r
  Informatik, 2022.

\bibitem{CDFW23}
O.~Carton, G.~Dou{\'{e}}neau{-}Tabot, E.~Filiot, and S.~Winter.
\newblock Deterministic regular functions of infinite words.
\newblock In K.~Etessami, U.~Feige, and G.~Puppis, editors, {\em 50th
  International Colloquium on Automata, Languages, and Programming, {ICALP}
  2023, July 10-14, 2023, Paderborn, Germany}, volume 261 of {\em LIPIcs},
  pages 121:1--121:18. Schloss Dagstuhl - Leibniz-Zentrum f{\"{u}}r Informatik,
  2023.

\bibitem{Cayley1889}
A.~Cayley.
\newblock A theorem of trees.
\newblock 23:376--378, 1889.

\bibitem{Cou94}
B.~Courcelle.
\newblock Monadic second-order definable graph transductions: {A} survey.
\newblock {\em Theor. Comput. Sci.}, 126(1):53--75, 1994.

\bibitem{DFJL17}
L.~Dartois, P.~Fournier, I.~Jecker, and N.~Lhote.
\newblock On reversible transducers.
\newblock In I.~Chatzigiannakis, P.~Indyk, F.~Kuhn, and A.~Muscholl, editors,
  {\em 44th International Colloquium on Automata, Languages, and Programming,
  {ICALP} 2017, July 10-14, 2017, Warsaw, Poland}, volume~80 of {\em LIPIcs},
  pages 113:1--113:12. Schloss Dagstuhl - Leibniz-Zentrum f{\"{u}}r Informatik,
  2017.

\bibitem{DBLP:conf/lics/DartoisGGK22}
L.~Dartois, P.~Gastin, R.~Govind, and S.~N. Krishna.
\newblock Efficient construction of reversible transducers from regular
  transducer expressions.
\newblock In C.~Baier and D.~Fisman, editors, {\em {LICS} '22: 37th Annual
  {ACM/IEEE} Symposium on Logic in Computer Science, Haifa, Israel, August 2 -
  5, 2022}, pages 50:1--50:13. {ACM}, 2022.

\bibitem{DBLP:conf/lics/DartoisGK21}
L.~Dartois, P.~Gastin, and S.~N. Krishna.
\newblock Sd-regular transducer expressions for aperiodic transformations.
\newblock In {\em 36th Annual {ACM/IEEE} Symposium on Logic in Computer
  Science, {LICS} 2021, Rome, Italy, June 29 - July 2, 2021}, pages 1--13.
  {IEEE}, 2021.

\bibitem{DJR16}
L.~Dartois, I.~Jecker, and P.~Reynier.
\newblock Aperiodic string transducers.
\newblock In {\em Developments in Language Theory - 20th International
  Conference, {DLT} 2016, Montr{\'{e}}al, Canada, July 25-28, 2016,
  Proceedings}, pages 125--137, 2016.

\bibitem{DFKL22}
V.~Dave, E.~Filiot, S.~N. Krishna, and N.~Lhote.
\newblock Synthesis of computable regular functions of infinite words.
\newblock {\em Log. Methods Comput. Sci.}, 18(2), 2022.

\bibitem{DGK-lics18}
V.~Dave, P.~Gastin, and S.~N. Krishna.
\newblock {Regular Transducer Expressions for Regular Transformations}.
\newblock In M.~Hofmann, A.~Dawar, and E.~Gr{\"a}del, editors, {\em
  {P}roceedings of the 33rd {A}nnual {ACM\slash IEEE} {S}ymposium on {L}ogic
  {I}n {C}omputer {S}cience ({LICS}'18)}, pages 315--324, Oxford, UK, July
  2018. ACM Press.

\bibitem{EH01}
J.~Engelfriet and H.~J. Hoogeboom.
\newblock M{SO} definable string transductions and two-way finite-state
  transducers.
\newblock {\em ACM Transactions on Computational Logic}, 2(2):216--254, 2001.

\bibitem{Perrin-Pin-Infinite-words}
D.~Perrin and J.-E. Pin.
\newblock {\em Infinite Words: Automata, Semigroups, Logic and Games}, volume
  141.
\newblock Elsevier, 2004.

\bibitem{Shepherdson59}
J.~C. Shepherdson.
\newblock The reduction of two-way automata to one-way automata.
\newblock {\em {IBM} J. Res. Dev.}, 3(2):198--200, 1959.

\end{thebibliography}

\end{document}